\documentclass[%
 reprint,
superscriptaddress,
 amsmath,amssymb,
 aps,
pra,
]{revtex4-2}

\usepackage{tikz} 
\usepackage{amsmath,amsfonts,amsthm, mathtools}
\usepackage{float}
\usepackage{dsfont}
\usepackage{csquotes}
\usepackage{amssymb}
\usepackage{verbatim}
\usepackage{mathtools}
\usepackage{rotating}
\usepackage{thmtools, thm-restate}
\usepackage{bm}
\usepackage{hyperref}
\usepackage{mathtools}
\usepackage{mathbbol}
\usepackage{pifont}
\usepackage{appendix}
\usepackage{tocloft}
\usepackage{titletoc}
\newcommand{\cmark}{\ding{51}}%
\newcommand{\xmark}{\ding{55}}%
\definecolor{darkred}  {rgb}{0.5,0,0}
\definecolor{darkblue} {rgb}{0,0,0.5}
\definecolor{darkgreen}{rgb}{0,0.5,0}
\hypersetup{
  colorlinks = true,
  urlcolor  = blue,         
  linkcolor = darkblue,     
  citecolor = darkgreen,    
  filecolor = darkred       
}
\theoremstyle{definition}

\newtheorem{corollary}{Corollary}
\newtheorem{definition}{Definition}
\newtheorem{conjecture}{Conjecture}
\newtheorem{lemma}{Lemma}
\newtheorem{proposition}{Proposition}
\newtheorem{theorem}{Theorem}

\newtheorem*{remark}{Remark}

\newcommand{\mbb}{\mathbb}
\newcommand{\mc}{\mathcal}

\newcommand{\tr}{\textrm{Tr}}

\usepackage{multirow}
\newcommand{\ket}[1]{|#1\rangle}
\newcommand{\bra}[1]{\langle #1|}
\newcommand{\op}[2]{|#1\rangle\langle#2|}

\definecolor{cool_green}{rgb}{0.0, 0.5, 0.0}
\newcommand{\yujie}{\color{black}}

\newcommand{\blk}{\color{black}}

\usepackage{blkarray}
\usepackage{adjustbox}
\usepackage{graphicx}
\usepackage{dcolumn}
\usepackage{bm}

\begin{document}

\preprint{APS/123-QED}

\title{Criteria for optimal entanglement-assisted long baseline telescopy}

\author{Yujie Zhang}
\email{yujie4physics@gmail.com}
\affiliation{Institute for Quantum Computing and Department of Physics \& Astronomy,
University of Waterloo, 200 University Ave W, Waterloo, Ontario, N2L 3G1, Canada}
\affiliation{Perimeter Institute for Theoretical Physics, 31 Caroline Street North, Waterloo, Ontario, Canada N2L 2Y5}
\author{Thomas Jennewein}
\affiliation{Institute for Quantum Computing and Department of Physics \& Astronomy,
University of Waterloo, 200 University Ave W, Waterloo, Ontario, N2L 3G1, Canada}
\affiliation{Department of Physics, Simon Fraser University, 8888 University Dr W, Burnaby, BC V5A 1S6, Canada}
\date{\today}
\begin{abstract}

    Entanglement-assisted telescopy protocols have been proposed as a means to extend the baseline of optical interferometric telescopes. However, the optimal entangled resource and a clear optimality criterion have remained unclear. Here, we propose a novel framework for systematically characterizing entanglement-assisted telescopy by integrating quantum metrology tools with the superselection rule (SSR) framework from quantum information theory. In our approach, the estimation problem in quantum telescopy is rigorously quantified using the quantum Fisher information (QFI) under SSR constraints. Building on this framework, we derive the fundamental limits of astronomical parameter estimation with finite entanglement resources and introduce new protocols that outperform previous methods and asymptotically saturate the optimal bound. Moreover, our proposed protocols are compatible with existing linear-optical technology and could inspire practical quantum telescopy schemes for near-term, lossy, and repeaterless quantum networks.
\end{abstract}

\maketitle

\section{Introduction}
The basic principles of stellar interferometry involve the coherent measurement of light using distinct collection telescopes to form an effective large-scale imaging system. This method enhances the angular resolution well beyond the diffraction limit of a single telescope, restricted by its aperture~\cite{Cittert1934, Zernike1938}. Unlike interferometric imaging in the radio-frequency spectrum, where Earth-sized telescope arrays have been highly successful, for example, in imaging black holes~\cite{thompson2017, Akiyama2019a}, optical systems face significant scaling limitations.  In the standard `direct detection' approach, photons collected at two telescopes are coherently recombined for interference~\cite{Monnier2003, shao1992}. However, this method requires a stable optical channel between the telescopes, and current technologies achieve a baseline of approximately 300 meters~\cite{Monnier2003}. Furthermore, methods based on optical homodyning using local oscillators have fundamental difficulties in achieving high signal-to-noise ratios in this regime~\cite{lawson2000}.  These challenges have motivated interferometric telescopy assisted by quantum network technologies~\cite{Awschalom2021,   Simon2017, huang2025}. \par   

The first quantum-networking-enhanced optical interferometry method was proposed by Gottesman, Jennewein, and Croke (GJC)~\cite{Gottesman2012}, who showed that high signal-to-noise ratio local measurements can be achieved at arbitrarily long baselines by utilizing single-photon-entangled (SPE) states provided by quantum repeater networks. This surprising result was rigorously explained by Tsang~\cite{tsang2011} using tools from quantum optics and quantum estimation theory. It was recently demonstrated experimentally~\cite{Matthew2023, Diaz21}, where it was highlighted that the \textit{nonlocal} (NL) or entangled property of the ancilla single-photon states is essential for the enhancement in quantum telescopy. Building on this, several entanglement-assisted protocols have been developed, such as using quantum memories and error correction to reduce the number of required single-photon entangled states~\cite{Khabiboulline2019, Huang2022},  implementing control gates to boost the success rate of the original GJC protocol~\cite{Czupryniak2023}, or using multiple copies of SPE states to achieve optimal performance~\cite{Marchese2023, Czupryniak2022} asymptotically. Recently, continuous-variable quantum teleportation using two-mode squeezed vacuum states has shown promising results~\cite{wang2023astronomical, huang2024}. Despite these advances, a clear, operational \emph{criterion} for when a given entanglement-assisted scheme is optimal under the physical constraints relevant to long-baseline optics has been lacking. Here, we provide a systematic study of entanglement-assisted protocols, recognizing that both \textit{locality} and \textit{superselection rule} (SSR) constraints must be taken into account to quantify the performance of different entanglement-assisted protocols clearly. \par 

The locality constraint, as introduced in~\cite{tsang2011}, is a well-recognized limitation in classical astronomical imaging operating without direct detection. While entangled (‘nonlocal’) resources can overcome this constraint, studying locality alone is insufficient for comparing different quantum telescopy protocols. A previously overlooked yet critical constraint arises from the absence of a phase reference -- also known as the superselection rule (SSR)~\cite{Barlett2007}. By examining the quantum Fisher information of astronomical photon states under the SSR constraint, we show that an SSR associated with $U(1)$ symmetry~\cite{marvian2012, Gour2009} places additional limitations on the amount of information extractable from astronomical interference measurements. To address this limitation, we systematically analyze and utilize various ancilla states as a \textit{phase reference} (PR) to effectively circumvent the SSR constraint. 

This framework enables us to identify the fundamental limits of entanglement-assisted schemes, quantify the utility of different ancilla states~\cite{Marchese2023, wang2023astronomical}, and design improved protocols. Concretely, in Section~\ref{sectionII} we review quantum telescopy as a quantum parameter-estimation problem and the photon-number $U(1)$ superselection rule (SSR) relevant to long-baseline optical telescopy. In Section~\ref{sectionIII} we analyze the implications of the local SSR, derive the general limitations it imposes on any entanglement-assisted scheme (tight QFI bounds), and provide criteria for ancilla-assisted schemes that overcome these limitations. In Section~\ref{sectionIV}, we evaluate and compare existing ancilla states used in quantum telescopy protocols within this framework and introduce a new entanglement-assisted ancilla that outperforms prior approaches under the same constraints. In Section~\ref{sectionV} we present a linear-optics, local-operations-and-classical-communication (LOCC) implementation that achieves the optimal quantum-telescopy limit by saturating the QFI bound derived here.
\blk
\par 

\section{Preliminaries}
\label{sectionII}
\subsection{Quantum Telescopy}
The optical signal received by the two telescopes can be modeled as a bipartite, two-mode, quasi-monochromatic weak thermal source with mean photon number $\epsilon \ll 1$, as is typical in optical interferometry~\cite{Mandel1995}. The quantum state of the collected light can thus be described by the following density operator:
 \begin{equation}\label{eq: rho}
\rho_s=(1-\epsilon)\rho_s^{(0)}+\epsilon\rho_s^{(1)}+O(\epsilon^2)\rho_s^{(>1)},
 \end{equation}
where the single-photon term $\rho_s^{(1)}$, which is the leading-order contribution encoding the complex visibility $g$, can be expressed in the single-photon two-mode Fock basis $\{\ket{0}_A\ket{1}_B, \ket{1}_A\ket{0}_B\}$ as:
\begin{equation}\label{eq: rho_1}
\rho_s^{(1)}=\frac{1}{2}\begin{pmatrix}
1 & g \\
g^* & 1
\end{pmatrix},
\end{equation}
with the two telescopes labeled as $A$ and $B$ (Fig.~\ref{fig:telescope1}), comprising an optical interferometer. \par 
Since the intensity distribution of an astronomical object relates to the complex visibility $g=|g|e^{i\theta}$ through the Van Cittert-Zernike theorem~\cite{Zernike1938, Cittert1934}, achieving a precise estimation of $g$ is the central objective of an astronomical interferometric imaging problem~\cite{Monnier2003}. \par 
\begin{figure}[t]
    \centering
\includegraphics[width=0.45\textwidth]{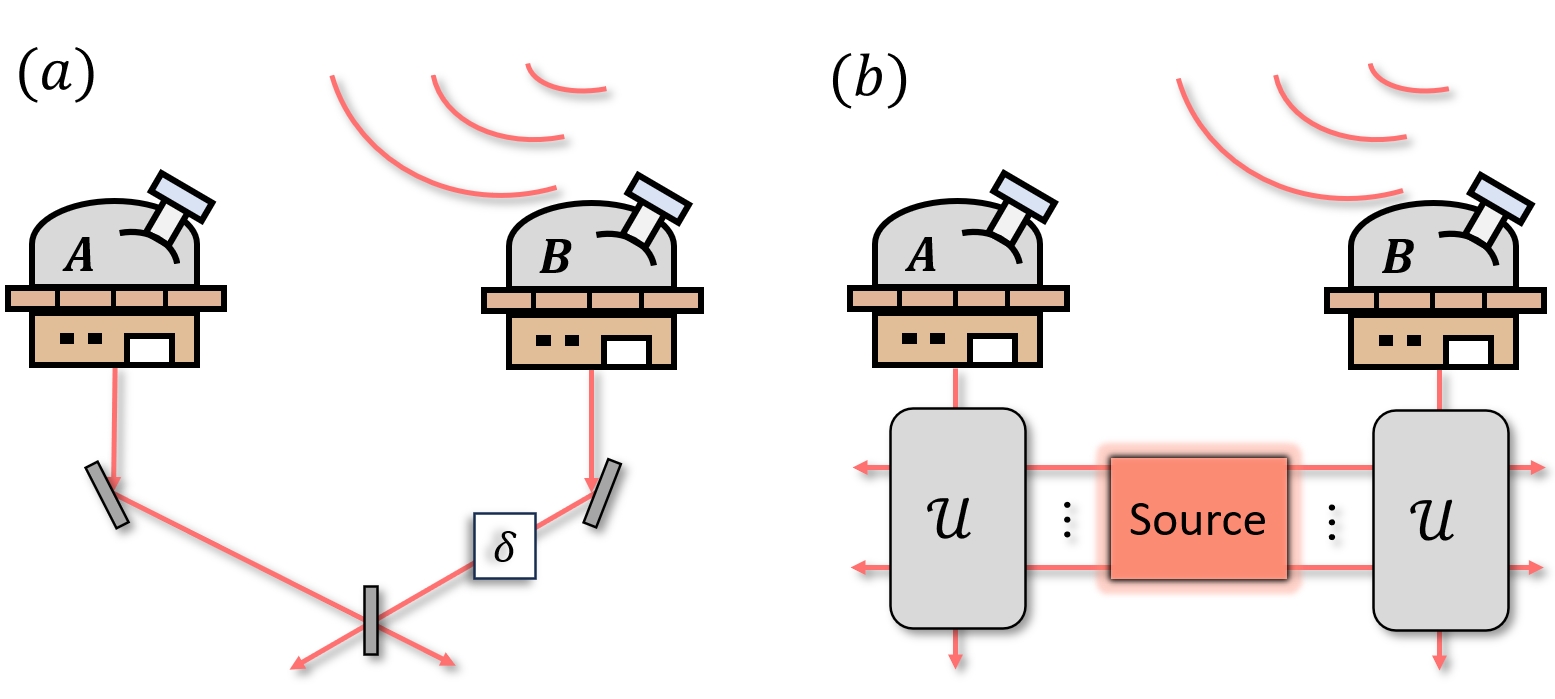}
    \caption{(a): Direct detection: light collected at two telescopes A and B is combined to interfere. (b): Entanglement-assisted protocol: shared quantum resources are provided to each telescope to enable interferometric imaging via local linear-optical circuit $\mc{U}$.}
    \label{fig:telescope1}
\end{figure}
To quantify the amount of information about $g$ contained in the source state, we consider the quantum Fisher information (QFI) matrix~\cite{Braunstein1994} $\mbb{H}_{\mu,\nu}$ of estimating parameter $\mu,\nu\in (|g|, \theta)$ given the source state $\rho_s$.  The QFI matrix upper bounds the classical Fisher information (FI) matrix $\mathbb{F}$ for any measurement, and via the quantum Cramér–Rao bound it lower-bounds the estimation covariance: $\mbb{H}^{-1}\le\mbb{F}^{-1}\le\mathrm{Cov}$. In the subsequent discussion, we will refer to  $\mathbb{F}_{\mu}$ and $\mathbb{H}_{\mu}$ as the diagonal elements for parameter $\mu$.

Unlike previous works~\cite{tsang2011, Czupryniak2022, wang2023astronomical, huang2024}, which used Fisher information $\mbb{F}$ as the metric for quantifying quantum telescopy protocols, we begin by analyzing the quantum Fisher information $\mbb{H}$ associated with the joint state of the source state $\rho_s$ and ancilla state $\rho_a$ (introduced below).

\begin{theorem}
\textit{Optimal quantum telescopy:} the QFI of estimating $|g|$ and $\theta$ in source state $\rho_s$ is given by:
\begin{equation}
\label{eq: telescope QFI}
\mbb{H}_{|g|}[\rho_s]=\frac{\epsilon}{1-|g|^2}+O(\epsilon^2),~\mbb{H}_{\theta}[\rho_s]=|g|^2\epsilon+O(\epsilon^2).
\end{equation}
These expressions are derived in the Appendix~\ref{appendixA}, assuming that no constraint is present and therefore serve as an upper bound for any quantum telescopy protocol. 
\end{theorem}
The QFI above is attainable in idealized direct interference with a lossless optical channel~\cite{tsang2011}. In the long-baseline setting, however, SSR constraints arise when direct interference is unavailable, making an explicit phase reference essential. Rather than overlooking this critical constraint, we provide a detailed analysis of the impact of SSR and phase references. 

\par 
\subsection{Superselection Rule}

In what follows, we briefly review the formulation of photon-number superselection rules ($U(1)$-SSR) as presented in~\cite{Barlett2007} and provide further details in Appendix~\ref{appendixB}. 

\textit{Global photon-number SSR--} In quantum optics experiments, the quantum states of optical modes are always referred to a phase reference. The absence of an absolute phase reference imposes constraints on the types of states that can be prepared and measured~\cite{Barlett2007}. 

Consider a quantum state of $K$ different modes defined relative to a phase reference $\ket{\psi}=\sum c_{n_1,\cdots,n_K}\otimes_{i=1}^K\ket{n_i}$ with $\{\otimes_{i=1}^K\ket{n_i}_i\}$ being the multimode Fock state basis on Hilbert space $\mc{H}$, and $\ket{n_i}_i$ denoting $n_i$ photons in mode $i$. When an absolute phase reference is absent, it is operationally equivalent to replace the state by its block-diagonal representative with respect to the Hilbert space decomposition $\mc H = \bigoplus_{n=0}^{\infty} \mc H_n$, where $\mc H_n:=\text{Span}\{\otimes_{i=1}^K\ket{n_i}_i:\sum_i n_i=n\}$ denotes the subspace with total photon number $n$~\cite{Barlett2007}. Under this constraint, the state is replaced by a density operator that is block-diagonal in the total photon-number subspaces:
\begin{equation}
\mc E_\text{g-ssr}(\op{\psi}{\psi})=\sum_{n=0}^{\infty} P_n\op{\psi}{\psi}P_n:=\bigoplus_n \sigma_n,
\label{eq:gl-ssr}
\end{equation}
where $P_n=\underset{\sum_i n_i=n}{\sum}\otimes_{i=1}^K\op{n_i}{n_i}_i$ acts on the joint multimode state, and projects it onto the $n$-photon subspace $\mc{H}_n$ with $\sigma_n$ being the subnormalized state on $\mc H_n$. For the simplest single-mode case with $\ket{\psi}=\sum c_{n}\ket{n}$, we have $\mc E_{\text{g-srr}}(\op{\psi}{\psi})=\sum_n^{\infty} |c_n|^2\op{n}{n}$. 

As a concrete example, the two-mode astronomical source in Eq.~\eqref{eq: rho} is already invariant under the global SSR, consistent with there being no absolute phase shared with the receivers.

\textit{Local photon-number SSR--} Beyond global SSR, when the system is bipartite with spatially separated sites A and B (e.g., telescopes in a long baseline optical interferometer), the absence of a correlated phase reference between the sites will further lead to \textit{local superselection rule}. In particular, one can take the decomposition of the underlying Hilbert space as~\cite{Bartlett2003, Schuch2004, Verstraete2003, Barlett2007}:
\begin{align}
\mc H=\mc H^A\otimes \mc H^B&=\bigoplus_{m=0}^{\infty} \bigoplus_{n=0}^{\infty} \mc H^A_n\otimes H^B_{m}. 
\label{eq:lssr-decom}
\end{align}
With a local SSR in force, any bipartite state $\rho^{AB}$ is operationally indistinguishable (under SSR-respecting processes) from its local-SSR-invariant counterpart. 
\begin{lemma}
\label{lem:SSR}
A bipartite state $\rho^{AB}$ subject to the photon-number local superselection rule can be represented as\cite{Verstraete2003}:
\begin{align}
\mc{E}_\text{l-ssr}(\rho^{AB})&=\sum_{m=0}^{\infty}\sum_{n=0}^{\infty}(P^A_{n}\otimes P^B_{m})\rho^{AB}(P^A_{n}\otimes P^B_{m}) \label{eq:l-ssr}\\
&:=\bigoplus_{n,m=0}^{\infty}\sigma_{n,m}\notag 
\end{align}
where $P^{A(B)}_n=\underset{\sum_i n_{i}=n}{\sum}\otimes_{i=1}^{K_{A(B)}}\op{n_i}{n_i}_{A_i(B_i)}$ acts on all local modes at site $A(B)$ and projects onto the $n$-photon subspace $\mathcal H^{A(B)}n$ (with $K_A$ and $K_B$ modes at $A$ and $B$ respectively). Thus $\mc E_{\text{l-ssr}}(\rho^{AB})$ is block-diagonal in local photon-number sectors with respect to Eq.~\ref{eq:lssr-decom}, with $\sigma_{n,m}$ being the (subnormalized) state on $\mathcal H^A_n\otimes\mathcal H^B_m$.  \par 

\end{lemma}

In summary, under a global SSR, only information within each total $n$-photon subspace is operationally accessible\footnote{Within a fixed 
$n$-photon subspace, there can be rich internal mode structure (e.g., spatial, temporal, polarization modes configuration)}, while coherence between different photon-number subspaces is effectively erased. For example, the relative phase in $\ket{0}+e^{i\phi}\ket{1}$ is unobservable. Consequently, any state may be represented by a density operator that is block-diagonal in photon-number subspaces, as in Eq.~\ref{eq:gl-ssr}.

Furthermore, in the presence of a local SSR for a bipartite system—which subsumes the global SSR—the constraint is stronger: coherence between different local photon-number sectors $(n,m)$ will be removed even within the same global total photon number subspace. For instance, the relative phase information in $\ket{0}_A\ket{1}_B+e^{i\phi}\ket{1}_A\ket{0}_B$ is inaccessible when parties $A$ and $B$ lack a correlated phase reference. Therefore, under a local SSR, any bipartite state can be replaced by its block-diagonal representative in the local photon-number subspaces as in Eq.~\ref{eq:l-ssr}.

Since local SSR subsumes the global SSR, we impose only the local SSR in what follows. 
\blk

\section{Quantum telescopy under Superselection Rule}  
\label{sectionIII}
To compare entanglement-assisted protocols on equal footing, we evaluate their performance via the quantum Fisher information (QFI) subject to the relevant local photon-number SSR. The QFI directly quantifies how well a candidate ancilla acts as a distributed phase-reference (PR) resource for entanglement-assisted quantum telescopy.
\blk

A first immediate consequence of Lemma~\ref{lem:SSR} is that, under the local SSR, the first-order term $\rho^{(1)}_s$ in Eq.~\ref{eq: rho} is fully dephased and carries no dependence on $g$:
\begin{equation}
\mc{E}_{\text{l-ssr}}(\rho_s) = (1-\epsilon)\mathbb{I}^{(0)} + \epsilon \frac{\mathbb{I}^{(1)}}{2} + O(\epsilon^2)\mathcal{E}_{\text{l-ssr}}(\rho_s^{(>1)}),
\end{equation}
where $\mbb I^{(n)}$ denotes the identity map on $\mc{H}_n$. Therefore, the achievable QFI in $\mc{E}_{\text{l-ssr}}(\rho_s)$ scales at most as $O(\epsilon^2)$, which explains why intensity interferometry --and any scheme without an explicit shared phase reference-- \blk is inefficient for long baseline astronomical imaging in the weak-thermal-sources regime~\cite{Brown1956, Monnier2003}.

To overcome this limitation, one may supply a shared phase-reference (PR) \cite{Barlett2006} in the form of a distributed ancilla state $\rho_a^{AB}$. This resource serves as a phase reference to the astronomical source $\rho_s^{AB}$. And formally, the joint state, subject to the local SSR, is:
\begin{align}
\mc E_{\text{l-ssr}}(\rho^{AB})= \sum_{n, m=0}^{\infty} (P^A_n \otimes P^B_m)[\rho_s^{AB} \otimes \rho_a^{AB}](P^A_n \otimes P^B_m), \label{eq: ancilla-state}
\end{align}
Here we note $P^A_n$ and $P^B_m$ project onto the full local n- and m-photon subspaces of the joint system at each site. Each site contains one source mode and, depending on $\rho_a^{AB}$, possibly many ancilla modes (see Examples in Section~\ref{sectionIV}). \blk

Since $\rho_a^{AB}$ and $\mc{E}_{\text{l-ssr}}$ are $g$-independent, by additivity and data-processing (monotonicity) of the QFI, we have:
\begin{lemma}
The QFI of the composite system in the presence of SSR is upper-bounded by the optimal QFI
\label{lem:QFI-bound}
\begin{subequations}
\begin{align}
&\mbb{H}_{|g|}[\mc{E}_{\text{l-ssr}}(\rho^{AB})] \le \mbb{H}_{|g|}[\rho_s^{AB}]\\
&\mbb{H}_{\theta}[\mc{E}_{\text{l-ssr}}(\rho^{AB})] \le \mbb{H}_{\theta}[\rho_s^{AB}],
\end{align}
\end{subequations}
\end{lemma}
And we can define the following quantity:
\begin{definition}
\label{def:QFI ratio}
\textit{QFI ratio} is defined by the ratio of its QFI to the optimal QFI as:
\begin{equation}
\mathbb{h}[\rho^{AB}_a]= \min \{\frac{\mbb{H}_{|g|}[\mc{E}_{\text{l-ssr}}(\rho^{AB})]}{\mbb{H}_{|g|}[\rho_s^{AB}]}, \frac{\mbb{H}_{\theta}[\mc{E}_{\text{l-ssr}}(\rho^{AB})]}{\mbb{H}_{\theta}[\rho_s^{AB}]}\}
\label{eq: QFI ratio}
\end{equation} 
\end{definition}

We emphasize that the \textit{QFI}-ratio serves only as an upper bound for the achievable \text{FI}-ratio in any actual protocol; their achievability will be discussed later via explicit protocols. 

To establish an upper bound on the QFI ratio for arbitrary entangled states, we write the most general pure ancilla as a superposition over states $\{\ket{n_A,m_B}\}_{n,m}$ with support on different total local photon-number subspaces, i.e.,
\begin{equation}
\ket{\psi}_a = \sum_{n,m\ge 0} f_{n,m}\,\ket{n_A,m_B},\qquad \ket{n_A,m_B}\in \mc H^A_n \otimes \mc H^B_m.
\label{eq:ancilla-state1}
\end{equation}
\yujie Here $\ket{n_A,m_B}$ denotes any unit vector supported on $\mc H^A_n \otimes \mc H^B_m$, that is, any state having $n$ photons at $A$ and $m$ photons at $B$, without specifying how they are distributed among local modes.

Thus $\ket{n_A,m_B}$ may represent
\begin{itemize}
    \item $\ket{n}_A\ket{m}_B$: a two-mode product state with $n$ photons in one mode at $A$ and $m$ photons in one mode at $B$ (or equivalently, $\ket{n}_A\ket{0}_A^{\otimes (n+m-1)} \ket{m}_B\ket{0}_B^{\otimes (n+m-1)}$);
    \item $\ket{1}_A^{\otimes n}\ket{0}_A^{\otimes m}\ket{0}_B^{\otimes n}\ket{1}_B^{\otimes m}$: a $2(n+m)$-mode product state with one photon per mode in the first $n$ modes at $A$, and one photon per mode in the last $m$ modes at $B$;
    \item $\sum_{\pi}\sqrt{p(\pi)}\,\pi\!\big(\ket{1}_A^{\otimes n}\ket{0}_A^{\otimes m}\ket{0}_B^{\otimes n}\ket{1}_B^{\otimes m}\big)$: a $2(n+m)$-mode entangled state, with $\pi$ ranging over permutations of local modes at $A$ and at $B$.
\end{itemize}
This mode-agnostic notion is particularly convenient for studying different ancilla states because any two vectors in $\mc H^A_n\otimes\mc H^B_m$ (i.e, two states with different internal mode configuration) are related by some unitary $U_{nm}^{AB}$ on $\mc{H}^A_n\otimes\mc H^B_m$, e.g., 
\begin{align}
&U_{nm}^{AB}\ket{n}_A\ket{0}_A^{\otimes n+m-1}\ket{m}_B\ket{0}_B^{\otimes n+m-1}\notag \\
=&\ket{1}_A^{\otimes n} \ket{0}_A^{\otimes m}\ket{0}_B^{\otimes n}\ket{1}_B^{\otimes m}
\end{align}
\blk
By Lemma~\ref{lem:ancilla-equivalence}, all such choices of internal mode configuration then yield the same QFI, so the QFI ratio below depends only on the amplitudes $f_{nm}$ of the ancilla state $\ket{\psi}_a$. 
\yujie
\begin{lemma}
\label{lem:ancilla-equivalence}
Let $\tilde\rho_a^{AB}=U\rho_a^{AB}U^\dagger$ where $U=\bigoplus_{n,m}U_{nm}^{AB}$ is block-diagonal on $\bigoplus_{n,m}\mc{H}^A_n\otimes\mc H^B_m$. Then, the QFI ratio $\mathbb{h}[\rho_a^{AB}]=\mbb{h}[\tilde\rho_a^{AB}]$.
\end{lemma}

A prove is given in  Appendix~\ref{appendixC}. Briefly, this is because the unitary $\bigoplus_{n,m}U_{nm}^{AB}$ commutes with the local SSR map $\mc{E}_{\text{l-ssr}}$. Since QFI is invariant under parameter-independent unitaries, we may therefore consider the ancilla in the mode-agnostic form of Eq.~\ref{eq:ancilla-state1}; a straightforward calculation to the first order then yields:
\blk
\begin{proposition}
\label{prop: assisted-QFI}
Given an ancilla state $\rho_a=\op{\psi}{\psi}_a$ in Eq.~\ref{eq:ancilla-state1}, the QFI for estimating $|g|$ and $\theta$ of the source state $\rho_s$ in the presence of an SSR is given by
\begin{align}
\mathbb{h}[\ket{\psi}_a]=\sum_{n,m=1}\frac{2|f_{n,m-1}|^2|f_{n-1,m}|^2}{|f_{n,m-1}|^2+|f_{n-1,m}|^2}    
\end{align}
Here, higher-order terms $O(\epsilon)$ are neglected for simplicity, so $\mathbb{h}=0$ should be interpreted as $\mathbb{h}=O(\epsilon)$.
\end{proposition}
\begin{remark}
If one first applies the global SSR to the ancilla by replacing $\rho_a$ with $\mc E_{\text{g-ssr}}(\rho_a)$ (its global SSR-invariant form), the QFI ratio in Proposition~\ref{prop: assisted-QFI} is unchanged. We do not impose the global SSR a priori because (i) the subsequent local SSR already subsumes it, and (ii) we wish to compare ancilla states that are not global-SSR-invariant (such as a two-mode squeezed vacuum state) to assess their usefulness for quantum telescopy.
\blk
\end{remark}

To facilitate quantitative discussion and comparison, we first consider the scenario where the ancilla state is restricted to at most $N$ photons:
\begin{theorem}
For ancilla state $\ket{\psi}_a$ with at most $N$-photon in total, we have:
\begin{align}
\mathbb{h}[\rho^{AB}_a]\le \cos(\frac{\pi}{N+2})\stackrel{\text{Large $N$}}{\approx} 1-\frac{\pi^2}{2N^2}
\end{align}
\label{thm: QFI-upper}
\end{theorem}
This upper bound is derived in Appendix~\ref{appendixC}, and it is noteworthy that the optimal ancilla states take the form $\ket{\psi}_a=\sum_{n=0}^Nf_{n}\ket{n_A,(N-n)_B}$ with $f_n\propto \sin\left(\frac{(n+1)\pi}{N+2} \right)$.
\begin{remark}
Asymptotically, under the mild assumption that $f_{n}$ above is 'smooth enough', we can prove the existence of a tighter upper bound:
\begin{equation}
\mathbb{h}[\rho^{AB}_a] \stackrel{\text{Large $N$}}{\approx} 1-\frac{\pi^2}{\langle N\rangle^2}+O(\frac{1}{\langle N\rangle^3})
\end{equation}
\end{remark}

Alternatively, one can also impose a constraint on the average photon number $\langle N\rangle = \sum_{n,m}|f_{n,m}|^2(n+m)$, for which a similar bound can be shown:
\begin{corollary}
\label{lem:QIF-upper}
For an ancilla state $|\psi\rangle_a$ with average photon number $\langle N\rangle$,
\begin{align}
\mathbb{h}[\rho^{AB}_a]\le \cos(\frac{\pi}{\langle N\rangle+2})\stackrel{\text{Large $\langle N\rangle$}}{\approx} (1-\frac{\pi^2}{2\langle N\rangle^2})
\end{align}
\end{corollary}

\par
\section{Comparing various protocols}
\label{sectionIV}
In the following, we provide explicit examples of different ancilla-assisted quantum telescopy protocols and their corresponding quantum Fisher information under SSR constraints. The performance of different ancilla states is summarized in Table~\ref{tab: assisted scheme}, and compared to the numerically optimal bound shown in Fig.~\ref{fig: QFI scaling}.

\yujie 
Throughout the section, the notation $\ket{n_A, m_B}$ denotes some unit vector supported on $\mc{H}^A_n\otimes\mc H^B_m$ whose internal mode configuration will be specified within each protocol.
\blk
\begin{table*}[t]
    \centering
        \renewcommand{\arraystretch}{1.2} 
    \setlength{\tabcolsep}{4pt} 
    \begin{tabular}{lllcccc}
    \hline
    \hline
     &\multirow{2}{*}{Ancilla state: $\ket{\psi}_a$}   &\multirow{2}{*}{QFI ratio: $\mathbb{h}$} &   \multicolumn{2}{c}{Resource\quad}  & \multicolumn{2}{c}{
 Realization\quad}\\
   &   & & NL  &\quad PR & LOCC\quad & Linear optics \\
           \hline
 GJC scheme\cite{Gottesman2012} & $\frac{1}{\sqrt{2}}(\ket{0}_A\ket{1}_B+\ket{1}_A\ket{0}_B)$      &   $\frac{1}{2}$   &  \cmark  & \cmark &  \cmark  &\quad  \cmark\\
 $N$-copy SPE\cite{Czupryniak2022,Marchese2023} &   $\frac{1}{\sqrt{2^N}}(\ket{0}_A\ket{1}_B+\ket{1}_A\ket{0}_B)^{\otimes N}$      &   $1-\frac{1}{N+1}$   &  \cmark  & \cmark &  \cmark  &\quad  \xmark\\
 KLM scheme&    $\frac{1}{\sqrt{N+1}}\sum\limits_{n=0}^N\ket{1}^{\otimes n}_A\ket{0}^{\otimes N-n}_A\ket{0}^{\otimes n}_B\ket{1}^{\otimes N-n}_B$      &   $1-\frac{1}{N+1}$   &  \cmark  & \cmark &  \cmark  &\quad  \cmark\\
Optimal-KLM &     $\sum\limits_{n=0}^N\frac{\sqrt{2}\sin(\frac{n+1}{N+2}\pi)}{\sqrt{N+2}}\ket{1}^{\otimes n}_A\ket{0}^{\otimes N-n}_A\ket{0}^{\otimes n}_B\ket{1}^{\otimes N-n}_B$      &   $\approx 1-\frac{\pi^2}{N^2}$   &  \cmark  & \cmark &  \cmark  &\quad  \cmark\\
CV-teleportation\cite{wang2023astronomical,huang2024} &  $\sum\limits_{n=0}^{\infty} \frac{\sqrt{2\langle N\rangle^{n}}}{\sqrt{(2+\langle N\rangle)^{n+1}}}\ket{n}_A\ket{n}_B\ket{\alpha}_A\ket{\alpha}_B$      &   $\ge 1-\frac{1}{\epsilon\langle N\rangle}^{\dagger}$   &  \cmark  & \cmark &  \cmark  &\quad  \cmark\\
Coherent state &  $\ket{\alpha}_A\ket{\alpha}_B$      &   $\approx 1-\frac{1}{2|\alpha|^2}$   &  \xmark  & \cmark &  \xmark  &\quad  \cmark\\
Two-photon (dual-rail) state  &  $\frac{1}{\sqrt{2}}(\ket{01}_A\ket{10}_B+\ket{10}_A\ket{01}_B)$      &   $0$   &  \cmark  & \xmark &  --  &\quad  -- \\
     \hline
     \hline
     \end{tabular}
    \caption{Summary of the QFI ratio $\mathbb{h}$ computed for various ancilla states used in quantum telescopy protocols. The ancilla state is classified as a nonlocal (NL) resource, a phase-reference (PR) resource, or both. The realizations of these QFI ratios are categorized according to whether they can be achieved using local operations with classical communication (LOCC) and linear optics. Importantly, not all nonlocal ancilla states can serve as a phase reference and overcome the SSR, and some classically correlated states can act as a phase reference but cannot overcome locality constraints. Therefore, a good ancilla state should provide both NL and PR resources. $^{\dagger}$: lower bound. }
    \label{tab: assisted scheme}
\end{table*}
\begin{figure}[t]
    \centering
\includegraphics[width=0.46\textwidth]{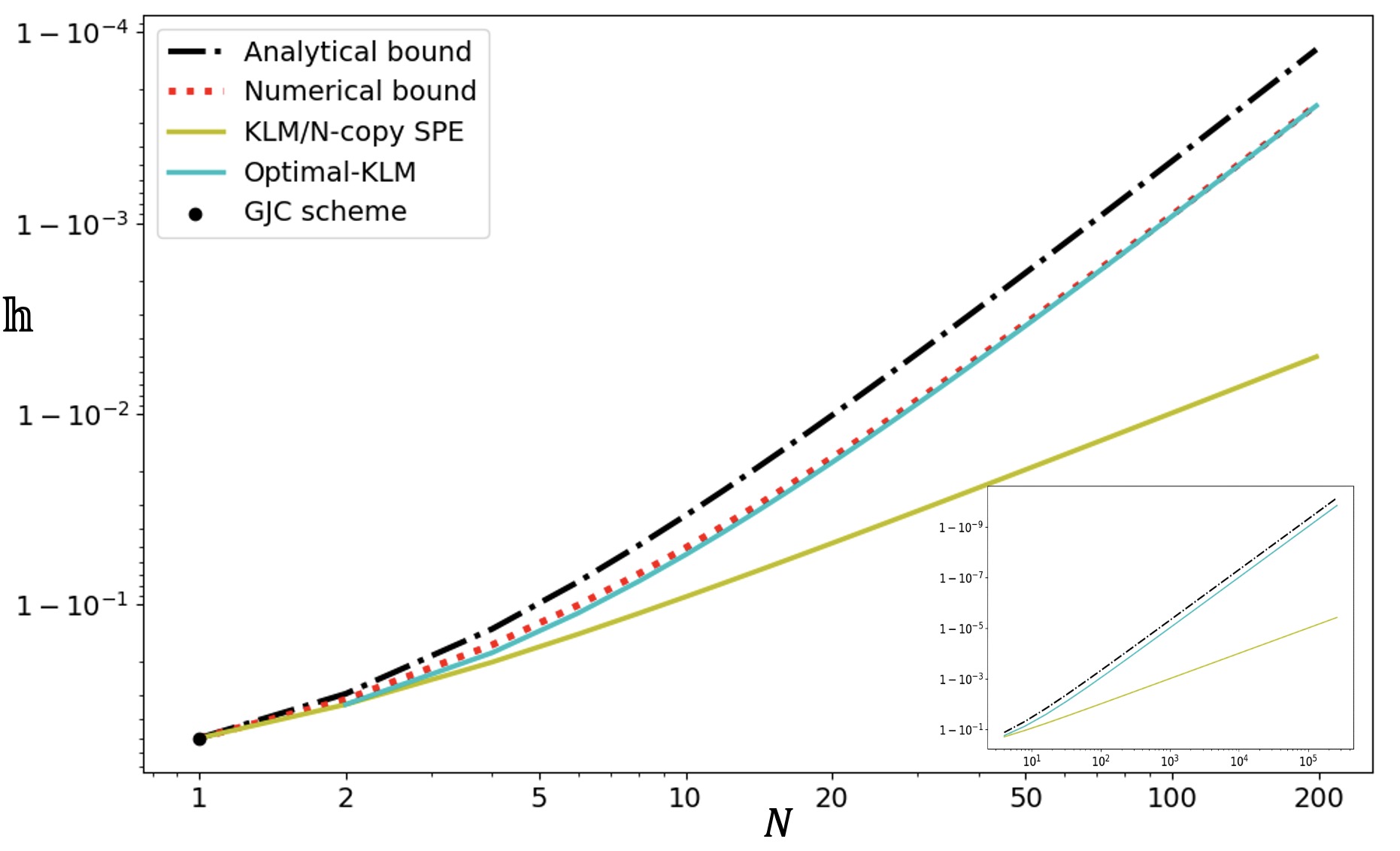}
    \caption{The Quantum Fisher information (QFI) ratio $\mathbb{h}$ as  a function of ancilla photon number $N$. Different entangled states are described in Table~\ref{tab: assisted scheme}}
    \label{fig: QFI scaling}
\end{figure}\par 

\subsection{GJC scheme} First, consider the original scheme from~\cite{Gottesman2012}, which uses the two-mode \blk single-particle entangled (SPE) ancilla state $\ket{\psi}^{\text{GJC}}_a = \frac{1}{\sqrt{2}}(\ket{0}_A\ket{1}_B + \ket{1}_A\ket{0}_B)$. Using Proposition~\ref{prop: assisted-QFI}, we can simply show:
\begin{equation}
\mathbb{h}[\ket{\psi}^{\text{GJC}}_a] = \frac{1}{2}.
\end{equation}
This QFI ratio is attainable via correlated interference measurements~\cite{Gottesman2012}, although the scheme only succeeds $\frac{1}{2}$ of the time. Importantly, this $\frac{1}{2}$ is a fundamental limit once the local SSR is correctly imposed, not an artifact of a particular protocol.

\subsection{$N$-copy SPE scheme—} A straightforward generalization of the GJC scheme considers the $N$-copy ancilla state~\cite{Czupryniak2022,Marchese2023} $(\ket{\psi}^{\text{GJC}}_a)^{\otimes N}$, which can be rewritten as:
\begin{align}
(\ket{\psi}^{\text{GJC}}_a)^{\otimes N}&= \frac{1}{\sqrt{2^N}}(\ket{0}_A\ket{1}_B + \ket{1}_A\ket{0}_B)^{\otimes N} \notag \\
&= \sum_{n=0}^{N} \sqrt{\frac{{N\choose n}}{2^N}} \ket{n_A, (N-n)_B}.
\end{align}
Note that the state $\ket{n_A,(N-n)_B}$ here represents a 2N-mode entangled states of the form:
\begin{equation}
  \ket{n_A,(N-n)_B}=\frac{1}{\sqrt{{N\choose n}}}\sum_{\pi}\pi(\ket{1}_A^{\otimes n} \ket{0}_A^{\otimes m}\ket{0}_B^{\otimes n}\ket{1}_B^{\otimes m}) 
\end{equation}
where $\pi$ represents different permutations on the $N$ local mode pairs on A and B.
\blk
Applying proposition~\ref{prop: assisted-QFI}, the QFI ratio is given by:
\begin{equation}
\mathbb{h}[(\ket{\psi}^{\text{GJC}}_a)^{\otimes N}] = \frac{N}{N+1}
\end{equation}
This matches the conjectured optimal Fisher information in~\cite{Czupryniak2022, Marchese2023}. However, as discussed in the Appendix~\ref{appendixDA}, while this QFI is theoretically achievable, whether it is realizable using linear optics alone (for general sources) remains open; existing claims implicitly assume operations that are not strictly linear-optical~\cite{Czupryniak2022, Marchese2023}.\par
\blk
\subsection{KLM-type scheme} We next introduce several quantum telescopy protocols inspired by the Knill–Laflamme–Milburn (KLM) approach to linear-optical quantum computation~\cite{Knill2001}. Additional details are provided in Appendix~\ref{appendixDB}.
 
The original KLM state used in \cite{Knill2001} is given by 
\begin{align}
    \ket{\psi}^{\text{KLM}}_a 
    &= \frac{1}{\sqrt{N+1}} \sum_{n=0}^N \ket{1}^{\otimes n}_A \ket{0}^{\otimes N-n}_A \ket{0}^{\otimes n}_B \ket{1}^{\otimes N-n}_B. \notag\\
    &=\sum_{n=0}^{N} \frac{1}{\sqrt{N+1}} \ket{n_A, (N-n)_B},  
\end{align} 
Different from the $N$-copy GJC state, here $\ket{n_A,(N-n)_B}$ represents a $2N$-mode product states. \blk Applying Proposition~\ref{prop: assisted-QFI} gives:
\begin{equation}
\mathbb{h}[\ket{\psi}_a^{\text{KLM}}]=\frac{N}{N+1},    
\end{equation}
Various modifications of KLM schemes have been discussed for various applications \cite{Franson2002}. In particular, we identify an \textit{Optimal-KLM} state below, which is asymptotically optimal for astronomical interferometric imaging:
\begin{align}
 \ket{\psi}^{\text{Opt}}_a &= \sum_{n=0}^N \frac{\sqrt{2} \sin\left(\frac{n+1}{N+2}\pi\right)}{\sqrt{N+2}} \ket{1}^{\otimes n}_A\ket{0}^{\otimes N-n}_A\ket{0}^{\otimes n}_B\ket{1}^{\otimes N-n}_B  \notag \\
    &=\sum_{n=0}^{N} \frac{\sqrt{2} \sin\left(\frac{n+1}{N+2}\pi\right)}{\sqrt{N+2}} \ket{1}^{\otimes n}_A \ket{n_A, (N-n)_B}
\end{align}
with an asymptotically optimal QFI ratio (see Fig.~\ref{fig: QFI scaling}):
\begin{equation}
\mathbb{h}[\ket{\psi}^{\text{Opt}}_a] \stackrel{\text{Large } N}{\approx} 1 - \frac{\pi^2}{N^2}.
\end{equation}

Unlike the $N$-copy SPE scheme, all KLM-type protocols are implementable with linear optics. Before proceeding, we briefly review several other relevant protocols of interest within the SSR framework.\par 
\subsection{CV protocol}
The continuous-variable (CV) teleportation protocol uses the two-mode squeezed vacuum state $\ket{\psi}^{\text{CV}}_a = \frac{1}{\cosh(r)} \sum_{n=0}^{\infty} \tanh^n r\ket{n}_A\ket{n}_B$ as the resource, 
where the squeezing parameter $r$ relates to the average photon number $\langle N\rangle = 2\sinh^2(r)$. 
\par 
At first glance, Proposition~\ref{prop: assisted-QFI} suggests $\mathbb{h}[\ket{\psi}^{\text{CV}}_a]=0$, which would appear to contradict the quantum enhancement demonstrated in \cite{wang2023astronomical,huang2024}. However, this apparent contradiction arises from a common oversight in the analysis of CV teleportation protocols: CV teleportation also requires a shared correlated coherent state $\ket{\alpha}_A\ket{\alpha}_B$ as a phase reference \cite{Furusama1998, Rudolph2001, Enk2001}. Thus, the true ancilla in the CV protocol is the combined state:
\begin{equation}
    \ket{\psi}^{\text{CV}}_a=\frac{1}{\cosh(r)}\sum\limits_{n=0}^{\infty} \tanh^n{r}\ket{n}_A\ket{n}_B\ket{\alpha}_A\ket{\alpha}_B,
\end{equation}
which, by Proposition~\ref{prop: assisted-QFI}, has nonzero QFI. In fact, as we show in Appendix~\ref{appendixDD}, we have:
\begin{equation}
\mathbb{h}[\ket{\psi}^{\text{CV}}_a] \geq \mathbb{f}[\ket{\psi}^{\text{CV}}_a] \stackrel{\text{Large }\langle N\rangle}{\approx} 1 - \frac{1}{\epsilon\langle N\rangle}.
\end{equation}
where $\mathbb{f}[\ket{\psi}^{\text{CV}}_a]$ is the Fisher information ratio defined similarly to $\mathbb{h}[\ket{\psi}^{\text{CV}}_a]$ in definition~\ref{def:QFI ratio}: 
\begin{align}
    \mathbb{f}[\rho^{AB}_a]= \min \{\frac{\mathbb{F}[\rho^{AB}_a]}{\mbb{H}_{|g|}[\rho_s^{AB}]}, \frac{\mathbb{F}[\rho^{AB}_a]}{\mbb{H}_{\theta}[\rho_s^{AB}]}\}
\end{align}
with $\mbb{F}_{|g|}$ representing the FI one can achieve using the CV protocol discussed in~\cite{wang2023astronomical}.

We note that here $\langle N\rangle$ denotes the mean photon number of the two-mode squeezed vacuum, and we assume, for simplicity, access to an unbounded correlated coherent reference, i.e., $|\alpha|^2\rightarrow \infty$. Since we are only providing an achievable lower bound on the QFI ratio here, this result does not contradict the correlated-coherent-states analysis in the following subsections. 
\blk

So far, we have discussed ancilla states for which the QFI ratio can asymptotically approach unity, since they serve as both nonlocal (NL) and phase reference (PR) resources~\cite{Barlett2006}. Next, we present two commonly confused examples of ancilla states and explain why they do not enable quantum-enhanced telescopy—namely, because they provide only a nonlocal resource or only a phase reference, but not both.

\subsection{Two-photon (dual-rail) entanglement}
Compared with single-particle entanglement, two-particle (dual-rail) entanglement is more common in quantum information\cite{Kok2007},  which can be explicitly written as a standard four-mode, two-photon entangled (dual-rail) state:
\begin{align}
    \ket{\psi}^{\text{TPE}}_a &= \frac{1}{\sqrt{2}}(\ket{1}_{A_1}\ket{0}_{A_2}\ket{0}_{B_1}\ket{1}_{B_2} + \ket{0}_{A_1}\ket{1}_{A_2}\ket{1}_{B_1}\ket{0}_{B_2})\\
    &=\ket{1_A,1_B}\in \mc{H}_1^A\otimes \mc H^B_1 \notag, 
\end{align}
which has exactly one photon at $A$ and one at $B$. By Proposition~\ref{prop: assisted-QFI}, the QFI ratio for this ancilla is zero. Intuitively, a phase reference requires `coherence' between {adjacent local-photon-number subspaces}, i.e., between $(n,m-1)$ and $(n-1,m)$ local-photon-number-subspaces, while the TPE state is confined to $(n,m)=(1,1)$. This distinction between nonlocal entanglement and a phase-reference resource has also been emphasized in the SSR literature (see Fig. 1 of Ref.~\cite{Schuch2004}).

Crucially, simply discarding a local mode does not “activate” the phase reference resource. Tracing out mode $A_2$ and $B_2$ produces an \emph{incoherent mixture} on $A_1B_1$,
$\tr_{A_2B_2}\big[\op{\psi}{\psi}|^{\text{TPE}}_a\big]
=\frac{1}{2}\big(|10\rangle\langle10|_{A_1B_1} + |01\rangle\langle01|_{A_1B_1}\big),$
not the single-particle entangled state; consequently, the QFI remains zero, agreeing with the fact that QFI is non-increasing under postprocessing. 

The same conclusion applies to other familiar ancilla states that serve purely as nonlocal resources. For example, the NOON state $\frac{1}{\sqrt{2}}(\ket{N}_A\ket{0}_B + \ket{0}_A\ket{N}_B)$ with $N\ge 2$ lacks the ability to function as phase references. Consequently, they also exhibit a vanishing QFI ratio and are unsuitable for achieving quantum-enhanced telescopy. 
\blk
\subsection{Correlated coherent state} 
Another relevant ancilla is the correlated coherent state $\ket{\psi}^{\text{Coh}}_a = \ket{\alpha}_A\ket{\alpha}_B$, which models synchronized local oscillators used in classical interferometric homodyne measurements~\cite{lawson2000}. For this state, the QFI ratio $\mathbb{h}[\ket{\psi}^{\text{Coh}}_a]$ is nonzero, indeed, as we evaluate in the Appendix~\ref{appendixDF}: 
\begin{equation}
 \mathbb{h}[\ket{\psi}^{\text{Coh}}_a]=1-\frac{1-e^{-2|\alpha|^2}}{2|\alpha|^2}\stackrel{\text{$|\alpha|^2\rightarrow \infty$}}{\approx} 1-\frac{1}{2|\alpha|^2},
\end{equation}
which approaches 1 when $|\alpha|\rightarrow \infty$. This, however, \emph{does not} mean the existence of a specific local measurement can saturate such QFI. Indeed, as we proved in the Appendix~\ref{appendixDF},  if $\rho_A^{T_A}>0$ (i.e., the state has a positive partial transpose), then for any local measurement strategy the attainable Fisher-information ratio vanishes in the weak-source limit. Therefore, our result does not contradict the fact that homodyne measurements cannot achieve quantum-enhanced telescopy~\cite{tsang2011}. 

This example, in fact,  illustrates that local measurements need not saturate the QFI in genuinely nonlocal estimation tasks~\cite{Zhou_2020}. 
Moreover, it demonstrates that a phase reference alone—such as that provided by $\ket{\alpha}_A\ket{\alpha}_B$—is insufficient; entanglement is needed to enable effectively nonlocal measurements that saturate the QFI.

\blk

\blk
\begin{figure}[t]
    \centering
\includegraphics[width=0.42\textwidth]{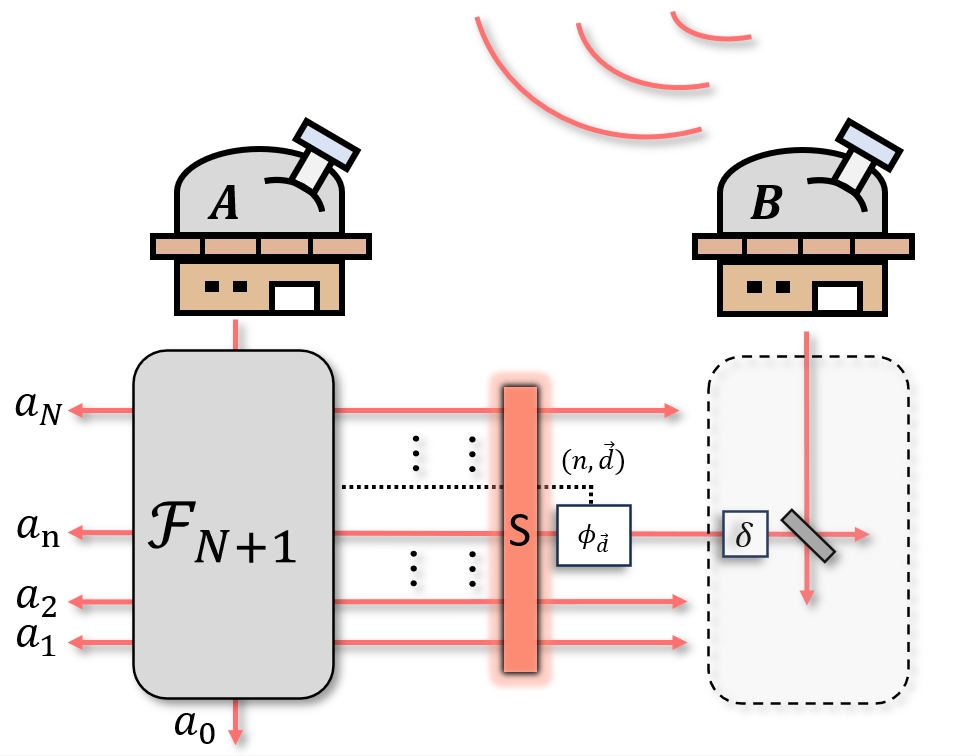}
    \caption{Schematic of near-deterministic teleportation: (1) Quantum Fourier transformation $\mc{F}_{N+1}$ is implemented at telescope A; (2) Measurement outcome $n$: number of photons detected, and $\vec{d}$: arrangement of detection are sent to telescope B for phase correction; (3) A scanning interferometric measurement is performed at telescope B. } 
    \label{fig:teleportation}
\end{figure}

\section{Towards repeater-based linear-optical implementation}
\label{sectionV}
We outline a linear-optical implementation of our protocol, with details provided in Appendix~\ref{appendixE} and illustrated in Fig.~\ref{fig:teleportation}. This approach assumes ideal (lossless) ancilla distribution, e.g., via a repeater network. It also differs from traditional passive schemes in that classical communication is permitted in the scheme below~\cite{Gottesman2012, Marchese2023}.
\begin{enumerate}
    \item A $2N$-mode bipartite ancilla state is shared between telescopes:
    \label{eq: tel-state}
$$\ket{\psi}_a=\sum_{n=0}^Nf_{n}\ket{1}^{\otimes n}_A\ket{0}^{\otimes N-n}_A\ket{0}^{\otimes n}_B\ket{1}^{\otimes N-n}_B $$
\item A quantum Fourier transform is performed on the $N+1$ modes at telescope A, followed by a photon-number-resolving measurement.

    \item The outcome (photon number $n$, configuration $\vec{d}$) is sent to B, which applies a phase shift $\phi_{\vec{d}} = \frac{2\pi}{N+1}\sum_i d_i$ to the $n$-th ancilla mode, where $\vec d$ encodes the detected output ports after the QFT (one entry per detected photon).
\item Telescope B combines the source mode with the 
$n$-th ancilla mode using a tunable beam splitter for interference measurement.
\end{enumerate}
Unlike the passive schemes studied in prior work \cite{Gottesman2012, Marchese2023}, this protocol requires a local delay line at telescope B to account for the latency associated with classical communication. Notably, recent advances, such as broadband all-optical quantum memory \cite{Nathan2024}, can provide robust optical delays of up to approximately 100 $\mu$s in a laboratory environment. These technological developments now make it feasible to implement quantum telescopy protocols involving a few rounds of local operations and classical communication. 

\section{Quantum Repeater Networks and Noisy Quantum Telescopy}
The advent of quantum repeaters has made the lossless distribution of entangled states across distant nodes in a quantum network an increasingly realistic prospect~\cite{Azuma2023, li2024generalized, Miguel2023}. As in the original quantum telescopy proposal~\cite{Gottesman2012}, quantum repeaters address two main challenges in long-baseline optical interferometry: phase noise from path fluctuations can be suppressed by active stabilization or entanglement distillation, and photon loss can be reduced with repeater protocols. These advances enable robust and scalable entanglement distribution, supporting reliable quantum telecopy protocols.

While future quantum networks can achieve lossless entanglement distribution, it remains crucial to study noisy, repeaterless protocols for near-term quantum networks, where transmission loss is significant. In this regime, the original Gottesman scheme offers no advantage over direct transmission, since its single-photon ancilla is equally vulnerable to loss. However, more complex, multiphoton entanglement-assisted protocols may be more robust, as shown in \cite{Marchese2023}.
\par 
Determining optimal ancilla states and measurements for noisy, repeaterless networks remains an open problem. Our formalism provides an upper bound on the achievable quantum Fisher information (QFI) for any quantum telescopy protocol, given the average surviving ancilla-photon number $\langle N \rangle$ (see Corollary~\ref{lem:QIF-upper}). For $\langle N \rangle \ll 1$,  the bound reduces to
$\mathbb{h}[\ket{\psi}^{\text{CV}}_a]\stackrel{\text{$\langle N\rangle \ll 1$}}{\approx} \frac{\pi}{4}\langle N \rangle $ \blk
A key experimental goal is to maximize the number of ancilla photons that survive the lossy channel, although this bound may not always be attainable. Once a high-QFI ancilla state is found, developing practical measurement schemes to approach this bound becomes essential.
\par 
\section{Conclusion} 
We revisited entanglement-assisted telescopy protocols by taking the superselection rule (SSR) into account, which, unlike locality constraints, has been broadly overlooked but is crucial for understanding and characterizing quantum telescopy protocols. As a result, any ancilla state must serve as both a phase reference and a nonlocal resource to overcome these limitations.

\par In this work, we also proposed new entanglement-assisted protocols that asymptotically achieve optimal quantum telescopy. These protocols are compatible with linear-optical implementations and are closely connected to linear-optical quantum computation (LOQC) protocols. Because of this intimate connection, recent graph-state–based LOQC advances offer an opportunity to inspire and develop new, noise-resilient quantum-telescopy schemes, especially once a global quantum network is in place~\cite{rudolph2017}.

Several open questions remain. In particular, it is crucial to investigate the design of practical, quantum-enhanced telescopy schemes for near-term networks, particularly when faced with real-world imperfections such as ground-to-space transmission losses and mode distinguishability errors. The framework presented here allows us to identify which ancilla resources are essential for such protocols by evaluating their quantum Fisher information, thereby guiding numerical searches for optimal ancilla states in lossy settings. Extending this approach to multi-telescope configurations -- each likewise subject to losses and imperfections -- poses another major challenge and is left for future work. Addressing these issues will be crucial for realizing near-term quantum telescopy in repeaterless networks.\par

\textit{Acknowledgments -- } 
The authors thank Virginia Lorenz, Yunkai Wang, Eric Chitambar, Raymond Laflamme, Pieter Kok and Robert Spekkens for helpful discussions. 
The research has been conducted at the Institute for
Quantum Computing, at the University of Waterloo,
which is supported by Innovation, Science, and Economic Development Canada. This research was supported with funds in part by NSERC Alliance Hyperspace, NSERC Discovery, and the Canada Excellence Research Chair (CERC) program.

\yujie
\section*{Note added}
After submitting our paper, an experimental paper appeared that demonstrated memory-assisted quantum telescopy~\cite{stas2025}, which is highly relevant to confirming our approach. This work shows that two-particle entanglement (TPE) and a phase-reference (PR) resource used \emph{together}, underpin quantum-memory–based protocols for quantum telescopy studied theoretically in~\cite{Khabiboulline2019, Khabiboulline2019b}.  

Concretely, in this scheme, a two-photon entangled state is first used to generate a remote entangled pair of quantum memories. Once the astronomical light interacts locally with the two memories, effectively, the first-order term of the source, Eq.~\eqref{eq: rho_1}, is mapped onto the memory degrees of freedom. Then, to complete the photon state storage to read out the information on the complex visibility $g$, a `$X$-basis measurements' is performed to erase the \emph{which-path} information~\cite{Khabiboulline2019, Khabiboulline2019b}. In practice, this is implemented using phase-locked, correlated coherent states (see Fig.~1e of~\cite{stas2025}). In our framework, these correlated coherent states supply the phase reference resource, while the entangled memories supply the required nonlocal resource. 

Thus, as exemplified by the recent experiment~\cite{stas2025}, even in the presence of quantum memories, achieving quantum-enhanced telescopy is constrained by local SSR and requires \emph{both} ingredients: (i) a phase reference resource and (ii) a nonlocal (entangled) resource.
\blk
\bibliography{ref}

\appendix
\setcounter{equation}{0}
\setcounter{lemma}{0}
\setcounter{proposition}{0}
\setcounter{theorem}{0}
\section{Thermal source state and optimal quantum Fisher information}
\label{appendixA}
In quantum optics, light from an astronomical object collected by two telescopes can be modeled as a two-mode quasi-monochromatic thermal source~\cite{tsang2011, Mandel1995}:
\begin{equation}
    \rho_s=\iint d^2\alpha d^2\beta\frac{\exp[-(\alpha^*, \beta^*)\Gamma^{-1}\begin{pmatrix}
\alpha \\
\beta
\end{pmatrix}]}{\pi^2\det\Gamma}\op{\alpha,\beta}{\alpha,\beta}
\end{equation}
where $\Gamma=\frac{\epsilon}{2}\begin{pmatrix}1 & g \\
g^*&  1 \end{pmatrix}$ and $g=|g|e^{i\theta}$ is the complex visibility, also known as mutual coherent function or first-order spatial coherence.
The parameter $\epsilon$, which quantifies the strength of the source, is typically small for astronomical sources at optical wavelengths~\cite{Monnier2003}. In the limit $\epsilon\ll 1$,  we can write the density operator of the state of the optical field on the image plane as
 \begin{equation}
\rho_s=(1-\epsilon)\rho_s^{(0)}+\epsilon\rho_s^{(1)}+O(\epsilon^2)\rho_s^{(>1)},
\label{eq: rho-app}
 \end{equation}
Where the leading-order term encoded with the complex visibility (coherence function) $g=|g|e^{i\theta}$ resides in the single-photon term $\rho_s^{(1)}$, which can be written in the telescope Fock basis $\{\ket{0}_A\ket{1}_B,\ket{1}_A\ket{0}_B\}$ as:
\begin{equation}\begin{aligned}\label{eq: rho_1-app}
\rho_s^{(1)}=&\frac{1}{2}[\ket{0}_A\ket{1}_B \bra{0}_A\bra{1}_B+\ket{1}_A\ket{0}_B\bra{1}_A\bra{0}_B\\
&+g\ket{0}_A\ket{1}_B\bra{1}_A\bra{0}_B+g^*\ket{1}_A\ket{0}_B\bra{0}_A\bra{1}_B] 
\end{aligned}
\end{equation}
The van Cittert-Zernike theorem \cite{Zernike1938} states that the visibility $g(\vec{u})$ (as a function of baseline $\vec{u}$) is the Fourier transform of the source distribution $I(\vec{x})$, i.e., $I(\vec{x})=\int d\vec{u} e^{-i2\pi \vec{u}\cdot\vec{x}}g(\vec{u})$ . Thus, the primary goal in astronomical interferometric imaging is to obtain an accurate estimate of the complex visibility for all baselines. 

In quantum‐estimation theory, the precision with which a parameter $\nu$ can be extracted from a quantum state $\rho_{\nu}$ is quantified by the quantum Fisher information (QFI) \cite{Braunstein1994}.  In particular, for the density operator given in Eq.~\ref{eq: rho_1-app}, the QFI is defined in terms of the symmetric logarithmic derivative (SLD) $\hat{L}_{\mu}$, where 
\begin{align}
\frac{\hat{L}_{\mu}\rho_s^{(1)}+\rho_s^{(1)}\hat{L}_{\mu}}{2}=\frac{\partial \rho_s^{(1)}}{\partial \mu}\quad\quad \mu\in\{|g|, \theta\} \notag  \\
\hat{L}_{|g|}=\frac{1}{1-|g|^2}\begin{pmatrix} -|g| &e^{i\theta}\\
e^{-i\theta}& -|g| \end{pmatrix},\hat{L}_{\theta}=|g|\begin{pmatrix} 0&ie^{i\theta}\\
-ie^{-i\theta}& 0 \end{pmatrix}
\end{align}
From these SLDs, the QFI metric of $\rho_s$ to first order in $\epsilon$ is $\hat{H}_{\nu,\mu}[\rho_s]=\epsilon\tr[\rho^{(1)}_s\frac{1}{2}\{\hat{L}_{\nu},\hat{L}_{\mu}\}]+O(\epsilon^2)$
\begin{subequations}
\begin{align}
\label{eq:telescope QFI}
&\mathbb{H}_{|g|}[\rho_s]=\frac{\epsilon}{1-|g|^2}+O(\epsilon^2) \\
&\mathbb{H}_{\theta}[\rho_s]=|g|^2\epsilon+O(\epsilon^2) \\
&\mathbb{H}_{\theta,|g|}[\rho_s]=O(\epsilon^2)
\end{align}
\end{subequations}
\blk
However, we note that this quantum Fisher information is not always attainable in practice. A well‐known obstacle is the locality constraint in quantum telescopy. As shown by Tsang \cite{tsang2011}, any protocol relying solely on local operations and classical communication (LOCC) cannot reach the optimal QFI. In this work, we focus instead on a ubiquitous constraint in quantum telescopy setups — the superselection rule (SSR). In the absence of a shared phase reference, both state preparation and measurement must respect the SSR‐imposed symmetry, as detailed below. This perspective lets us address questions that locality alone cannot resolve. For example, we compare entanglement-assisted protocols under SSRs and introduce schemes that, asymptotically, attain the optimal QFI with improved resource efficiency.

\section{Superselection rule, reference frames, and quantum estimation under SSR}
\label{appendixB}

In this section, we discuss the photon‐number superselection rule (SSR) encountered in interferometric imaging systems and its implications for the quantum estimation problem. We begin with a brief overview of reference frames and SSRs. In particular, we focus exclusively on the phase reference (PR) frame and the photon‐number $U(1)$-SSR. For a more comprehensive treatment, we refer the reader to Bartlett, Rudolph, and Spekkens \cite{Barlett2007}.
\begin{align}
    \text{no phase reference}\;\Longleftrightarrow\;\text{a $U(1)$ SSR is in force} \notag \\
    \text{supplying a phase reference}\;\Longleftrightarrow\;\text{lifting the SSR.} \notag
\end{align}

\textit{Global photon–number SSR:}
In a quantum optics experiment, the quantum states of optical modes are always defined relative to some phase references~\cite{Barlett2006}.   Consider a quantum state of $K$ different modes defined relative to a phase reference $\ket{\psi}=\sum c_{n_1,\cdots,n_K}\otimes_{i=1}^K\ket{n_i}$ with $\{\otimes_{i=1}^K\ket{n_i}_i\}$ being the Fock state basis on Hilbert space $\mc{H}$, where $\ket{n_i}_i$ represents $n_i$ photons in mode $i$.  Let 
$\hat N=\sum_{i=1}^{K}\hat N_i$ and $U(\phi)=e^{i\phi \hat N}$. Now consider another party, which has a different phase reference; the state $\ket{\psi}$ relative to the new phase reference frame will be given by the transformed state 
\begin{align}
U(\phi)\ket{\psi}=e^{i\phi\hat{N}}\ket{\psi}
\end{align}
where the unitary transformation $U(\phi)=e^{i\phi\hat{N}}$ depends on the number operator $\hat{N}$ and the angle $\phi$ between the two phase references and Alice's phase references.

Now consider the situation where the phase differences $\phi$ are entirely unknown, the lack of phase information restricts the kind of state prepared and measured, more rigorously, the state $\ket{\psi}$ can be effectively represented by a different quantum state obtained by averaging over an unknown global phase (“$U(1)$ twirling’’), which gives the global–SSR map:
\begin{align}
\label{eq:twirl-global}
\mathcal E_{\mathrm{g\mbox{-}ssr}}(\op{\phi}{\phi})
&=\int_0^{2\pi}\frac{d\phi}{2\pi}U(\phi)\op{\phi}{\phi} U^\dagger(\phi)\notag \\
&=\int_0^{2\pi}\frac{d\phi}{2\pi}\sum_{n,n'}e^{i\phi n}P_n\op{\phi}{\phi} P_{n'}e^{-i\phi n'} \notag \\
&=\sum_{n}P_n\op{\phi}{\phi} P_{n},
\end{align}
where the Hilbert space decompose as a direct sum of total–photon–number subspaces:
\begin{equation}
\label{eq:Hn-def}
\mathcal H=\bigoplus_{n=0}^{\infty}\mathcal H_n,~~
\mathcal H_n:=\mathrm{Span}\big\{\ket{n_1,\ldots,n_K}:\textstyle\sum_i n_i=n\big\},
\end{equation}
And the projector $P_{n}$ projects the state onto the entire $n$–photon subspace $\mc H_n$ across all $K$-modes (it is not a single mode projector), i.e.:
\begin{equation}
\label{eq:Pn-def}
P_n=\sum_{\sum_i n_i=n}\ket{n_1,\ldots,n_K}\bra{n_1,\ldots,n_K}.
\end{equation}

Thus, any state is operationally equivalent to its block–diagonal representative $\bigoplus_n P_n\rho P_n$.

\textit{Local photon–number SSR for two sites.} Similarly, consider a bipartite system with two spatially separated sites $A$ and $B$, for each bipartite state $\rho^{AB}$ (possibly multimode on each site), the absence of a \emph{correlated} phase reference between $A$ and $B$ implies that 
the invariant state in this case should be invariance under \emph{independent} local $U(1)$ shifts. Specifically, writing
$U_A(\phi_1)=e^{i\phi_1 \hat N_A}$ and $U_B(\phi_2)=e^{i\phi_2 \hat N_B}$ with
$\hat N_A$ ($\hat N_B$) the total local photon–number operator at $A$ ($B$). Under the local SSR, it is convenient to decompose the Hilbert space into total-local-photon-number subspaces: 
\begin{equation}
\label{eq:local-decomp}
\mathcal H^A\otimes\mathcal H^B
=\bigoplus_{n,m\ge 0}\mathcal H^A_n\otimes\mathcal H^B_m
=\bigoplus_{N=0}^{\infty}\bigoplus_{n=0}^{N}\mathcal H^A_n\otimes\mathcal H^B_{N-n}.
\end{equation}
And operationally, a bipartite state will be mapped by the local SSR twirling map:
\begin{lemma}
A bipartite state $\rho^{AB}$ under photon-number
superselection rules will be represented as~\cite{Verstraete2003}:
\begin{equation}
\mc{E}_{\text{l-ssr}}(\rho^{AB})=\sum_{n, m=0}^{\infty}(P^A_{n}\otimes P^B_{m})\rho^{AB}(P^A_{n}\otimes P^B_{m})
\end{equation}
where $P^A_n=\op{n}{n}$ represent projector onto $H^A_n$, and similarly for $P^B_n$
\label{lem:SSR}
\end{lemma}
\begin{proof}
A superselection rule imposed on the bipartite quantum state $\rho^{AB}$ implies a lack of shared phase references on the bipartite system. Since the phase references on systems A and B are uncorrelated, the local SSR implies a twirling map of the form:
\begin{align}
\label{eq:twirl-local}
&\mathcal E_{\mathrm{l\mbox{-}ssr}}(\rho^{AB})\notag \\
&=\int_0^{2\pi}\frac{d\phi_1d\phi_2}{(2\pi)^2}
\big(U_A(\phi_1)\otimes U_B(\phi_2)\big)\rho^{AB}
\big(U_A^\dagger(\phi_1)\otimes U_B^\dagger(\phi_2)\big)\notag \\
&=\sum_{n,m\ge 0}(P^A_n\otimes P^B_m)\rho^{AB}(P^A_n\otimes P^B_m),
\end{align}
where $\big(U_A(\phi_1)\otimes U_B(\phi_2)\big)\rho^{AB}
\big(U_A^\dagger(\phi_1)\otimes U_B^\dagger(\phi_2)\big)=\sum_{n,n',m,m'\ge 0}e^{i\phi_1(n-n')+i\phi_2(m-m')}(P^A_n\otimes P^B_m)\rho^{AB}(P^A_{n'}\otimes P^B_{m'})$ is the $U(1)$-twirling operation\cite{Barlett2007}, and $P^A_n$ ($P^B_m$) projects onto the \emph{entire} local $n$– ($m$–) photon subspace
$\mathcal H^A_n$ ($\mathcal H^B_m$). 
\end{proof}
Because independent local twirling maps include the global one, the local SSR subsumes the global SSR, i.e.
$\mathcal E_{\mathrm{l\mbox{-}ssr}}\circ\mathcal E_{\mathrm{g\mbox{-}ssr}}=\mathcal E_{\mathrm{l\mbox{-}ssr}}$. Henceforth, we impose only the local SSR.

As a direct consequence of the lemma above, 
applying the local SSR to the thermal source of Eq.~\eqref{eq: rho-app} gives
\begin{equation}
\label{eq:source-under-SSR}
\mathcal E_{\mathrm{l\mbox{-}ssr}}(\rho_s)
=(1-\epsilon)\mathbb I^{(0)}+\epsilon\frac{\mathbb I^{(1)}}{2}
+O(\epsilon^2)\mathcal E_{\mathrm{l\mbox{-}ssr}}\big(\rho_s^{(>1)}\big),
\end{equation}
where $\mbb I^{(n)}$ denotes the identity map on $\mc{H}_n=\oplus_{i=0}^n \mc{H}_i^A\otimes\mc{H}_{n-i}^B$. So the first–order term $\rho_s^{(1)}$ contribution becomes fully dephased and independent of $g$.
Hence, with no additional resource, the achievable QFI scales at most as $O(\epsilon^2)$, consistent with the inefficiency of intensity interferometry for weak thermal light~\cite{Brown1956, Monnier2003}.

However, it is always possible to use a shared bipartite ancilla state to build shared phase references and lift the restriction of SSR~\cite{Enk2005}. This idea is also studied and known as activating entanglement under SSR (see, e.g., Refs .~\cite{Enk2005, Bartlett2006a}). Simply put, by introducing an additional ancillary system $\rho_a$, the composite system of 
\begin{align}
    \rho^{AB}=\rho^{AB}_a\otimes \rho^{AB}_s
\end{align}
could contain more information about the complex visibility $g$ that is not erased by the superselection rules. And the physically accessible joint state is the locally twirled composite:
\begin{equation}
\label{eq:joint-under-SSR}
\mathcal E_{\mathrm{l\mbox{-}ssr}}(\rho^{AB})
=\sum_{n,m\ge 0}(P^A_n\otimes P^B_m)
\big[\rho^{AB}_s\otimes \rho^{AB}_a\big]
(P^A_n\otimes P^B_m).
\end{equation}
where $P^A_n$ and $P^B_m$ act on the joint multimode spaces at each site (source + ancilla).

Because $\rho^{AB}_a$ and $\mathcal E_{\mathrm{l\mbox{-}ssr}}$ do not depend on $g$, additivity and monotonicity of QFI under CPTP maps imply:

\begin{lemma}
\label{lem:QFI-bound-app}
For $\mu\in\{|g|,\theta\}$,
\begin{equation}
\mathbb H_{\mu}\left[\mathcal E_{\mathrm{l\mbox{-}ssr}}(\rho^{AB})\right]
\le
\mathbb H_{\mu}\left[\rho^{AB}_s\right].
\end{equation}
\end{lemma}

\begin{proof}
By additivity and monotonicity of  quantum Fisher information, we have for $\mu\in \{|g|,\theta\}$
\begin{equation}
\mbb{H}_{\mu}[\rho_{s}^{AB}]+\mbb{H}_{\mu}[\rho_{a}^{AB}] \ge  \mbb{H}_{\mu}[\rho_s^{AB}\otimes\rho_a^{AB}] \ge  \mbb{H}_{\mu}[E_{\mathrm{l\mbox{-}ssr}}(\rho^{AB})]\notag 
\end{equation}
Since $\mbb{H}_{\mu}[\rho_{a}^{AB}]$=0, we have the lemma proved. 
\end{proof}

It is therefore natural to quantify the usefulness of a given ancilla by the (dimensionless) QFI–ratio
\begin{equation}
\label{eq:h-ratio-app}
\mathbb h[\rho^{AB}_a]
:=\min\left\{
\frac{\mathbb H_{|g|}\left[\mathcal E_{\mathrm{l\mbox{-}ssr}}(\rho^{AB})\right]}
     {\mathbb H_{|g|}\left[\rho^{AB}_s\right]},\;
\frac{\mathbb H_{\theta}\left[\mathcal E_{\mathrm{l\mbox{-}ssr}}(\rho^{AB})\right]}
     {\mathbb H_{\theta}\left[\rho^{AB}_s\right]}
\right\}
\end{equation}
which takes value in $[0, 1]$ from lemma~\ref{lem:QFI-bound}. In the following, we will refer to this quantity as the QFI ratio and determine how to find the optimal entanglement ancilla state such that the QFI ratio can asymptotically approach one, thereby achieving optimal quantum telescopy. 
\blk
\section{Upper Bound on the Quantum Fisher Information Ratio. }
\label{appendixC}
\yujie 
As we discussed in the main text, in this section, we assume the most general ancilla state to be of the form:
\begin{equation}
\ket{\psi}_a = \sum_{n,m=0} f_{n,m} \ket{n_A, m_B},\quad \ket{n_A, m_B}\in\mc{H}^A_n\otimes\mc H^B_m.
\label{eq: ancilla-state1_app}
\end{equation}
Where $\ket{n_A, m_B}$ denotes any unit vector supported on $\mc{H}^A_n\otimes\mc H^B_m$, i.e., having $n$ photons at $A$ and $m$ photons at $B$, which can have different mode configurations such as $\ket{n}_A\ket{m}_B$ (or equivalently $\ket{n}_A\ket{0}_A^{\otimes n+m-1}\ket{m}_B\ket{0}_B^{\otimes n+m-1}$) and $\ket{n_A,m_B}=\ket{1}_A^{\otimes n} \ket{0}_A^{\otimes m}\ket{0}_B^{\otimes n}\ket{1}_B^{\otimes m}$. Since any choice of mode configurations is related by passive unitaries $U_{nm}^{AB}$ on $\mc{H}^A_n\otimes\mc H^B_m$, i.e., 
\begin{align}
&U_{nm}^{AB}\ket{n}_A\ket{0}_A^{\otimes n+m-1}\ket{m}_B\ket{0}_B^{\otimes n+m-1}\notag \\
&= \ket{1}_A^{\otimes n} \ket{0}_A^{\otimes m}\ket{0}_B^{\otimes n}\ket{1}_B^{\otimes m}
\end{align}
The following lemma allows us to discuss the QFI ratio of $\ket{\psi}_a$ in a mode-configuration-agnostic manner. 
\begin{lemma}
\label{lem:ancilla-equivalence}
Let $\tilde\rho_a^{AB}=U\rho_a^{AB}U^\dagger$ where $U=\bigoplus_{n,m}U_{nm}^{AB}$ is block-diagonal on $\bigoplus_{n,m}\mc{H}^A_n\otimes\mc H^B_m$. Then, the QFI ratio $\mathbb{h}[\rho_a^{AB}]=\mbb{h}[\tilde\rho_a^{AB}]$.
\end{lemma}
\begin{proof}
Let $\hat U:=\mathbb 1_{s}\otimes U$ act on the joint state
$\rho_s^{AB}\otimes\rho_a^{AB}$. Since $U$
preserves the ancilla’s local photon numbers, $\hat U$ should preserve the total (ancilla+source) local photon numbers and thus commute with the 
projectors: $ [\hat U, P^A_n\otimes P^B_m]=0$ for all $n,m$. Hence
\[
\mathcal E_{\mathrm{l\mbox{-}ssr}}\!\left(\rho_s^{AB}\otimes\tilde\rho_a^{AB}\right)
=\hat U\,\mathcal E_{\mathrm{l\mbox{-}ssr}}\!\left(\rho_s^{AB}\otimes\rho_a^{AB}\right)\hat U^\dagger.
\]
Since the quantum Fisher information is invariant under parameter–independent
unitaries, for each $\mu$,
\[
\mathbb H_\mu\!\left[\mathcal E_{\mathrm{l\mbox{-}ssr}}\!\left(\rho_s^{AB}\otimes\tilde\rho_a^{AB}\right)\right]
=
\mathbb H_\mu\!\left[\mathcal E_{\mathrm{l\mbox{-}ssr}}\!\left(\rho_s^{AB}\otimes\rho_a^{AB}\right)\right],
\]
and because the denominators $\mathbb{H}_{\mu}[\rho_s^{AB}]$ in the QFI–ratio in Eq~\ref{eq:h-ratio-app} do
not depend on the ancilla, both ratios (for $|g|$ and $\theta$) and hence their minimums are unchanged. Therefore
$\mathbb h[\rho_a^{AB}]=\mathbb h[\tilde\rho_a^{AB}]$.
\end{proof}
\blk
\begin{proposition}
\label{prop: assisted-QFI}
Given an ancilla state $\rho_a=\op{\psi}{\psi}_a$ of the form $|\psi\rangle_a = \sum_{n,m=0} f_{n,m} |n_A, m_B\rangle$, the quantum Fisher information in estimating $|g|$ and $\theta$ of source state $\rho_s$ in the presence of superselection rule is given by:
\begin{align}
&\mbb{H}_{|g|}[\mc{E}_{\text{l-ssr}}(\rho^{AB})]=\frac{\epsilon}{1-|g|^2}\sum_{n,m=1}\frac{2|f_{n,m-1}|^2|f_{n-1,m}|^2}{|f_{n,m-1}|^2+|f_{n-1,m}|^2}  \notag  \\
&\mbb{H}_{\theta}[\mc{E}_{\text{l-ssr}}(\rho^{AB})]=\epsilon|g|^2\sum_{n,m=1}\frac{2|f_{n,m-1}|^2|f_{n-1,m}|^2}{|f_{n,m-1}|^2+|f_{n-1,m}|^2}  \notag  \\
&\mathbb{h}[\rho^{AB}_a]=\sum_{n,m=1}\frac{2|f_{n,m-1}|^2|f_{n-1,m}|^2}{|f_{n,m-1}|^2+|f_{n-1,m}|^2}    
\end{align}
Here we neglect the higher term $O(\epsilon^2)$ for simplicity. Thus, in what follows, $\mathbb{h}=0$ actually means $\mathbb{h}=O(\epsilon^2)$.
\end{proposition}
\begin{proof}
Assuming the general form of the shared ancilla state:
\begin{equation}
\ket{\psi}_a = \sum_{n,m=0} f_{n,m} \ket{n_A, m_B},\quad \ket{n_A, m_B}\in\mc{H}^A_n\otimes\mc H^B_m,
\label{eq: ancilla-state1-app}
\end{equation}
with $\sum_{n,m=0}^{\infty}|f_{n,m}|^2=1$. By neglecting the higher order term from $\rho_s$, we have:
\begin{align}
 \mc{E}_{\text{l-ssr}}(\rho^{AB})&=\sum_{n,m}(P^A_{n}\otimes P^B_{m})\rho_s^{AB}\otimes\rho_a^{AB} (P^A_{n}\otimes P^B_{m})\notag \\
 &=\bigoplus_{n,m}\sigma^{AB}_{n,m}=\bigoplus_{n,m}\left[\sigma^{AB0}_{n,m}+\sigma^{AB
\epsilon}_{n,m}+O(\epsilon^2)\right]\notag
\end{align}
where $\sigma^{AB}_{n,m}$ are subnormalized. The zero-order term:
\begin{align}
  \sigma^{AB0}_{n,m}=(1 - \epsilon)|f_{n,m}|^2 \rho^{AB}_0\otimes\ket{n_A,m_B}\bra{n_A,m_B}  
\end{align}
is $g$-independent, and the first order term $\sigma^{AB\epsilon}_{n, m}$ can be expressed in the $2$-dimensional subspace spanned by $\{\ket{0}^s_A\ket{1}^s_B\ket{n_A,(m-1)_B}^a,\ket{1}^s_A\ket{0}^s_B\ket{(n-1)_A,m_B}^a\}$ (Here, we include superscript $s$ and $a$ for source and ancilla to avoid confusion):
\begin{subequations}
\begin{align}
\sigma^{AB\epsilon}_{n, m}&=\frac{\epsilon}{2}\begin{pmatrix}
{|f_{n,m-1}|^2}& {f^*_{n-1,m}f_{n,m-1}} g \\
{f_{n-1,m}f^*_{n,m-1}}g^* & {|f_{n-1,m}|^2}, \label{eq:sigmanm}
\end{pmatrix}
\end{align} 
\end{subequations}
Since $\sigma_{n,m}^{A,B}$ are density operators with disjoint support thus mutually orthogonal and $\sigma_{n,m}^{AB0}$ is $g$-independent, we have:
\begin{align}
\label{eq:QFI-app}
 \mbb{H}_{\mu}[\mc{E}_{\text{l-ssr}}(\rho^{AB})]&=\sum_{n,m=0}\mbb{H}_{\mu}[\sigma_{n,m}^{AB}]\quad\quad\quad \mu\in\{|g|,\theta\}.    \notag \\
 &=\sum_{n,m=0}\mbb{H}_{\mu}[\sigma_{n,m}^{AB\epsilon}]+O(\epsilon^2)
\end{align}
Writing $\sigma^{AB\epsilon}_{n, m}=\frac{(|f_{n,m-1}|^2+|f_{n-1,m}|^2)\epsilon}{2}\tilde \sigma^{AB\epsilon}_{n, m}$ with 
\begin{align}
\tilde \sigma^{AB\epsilon}_{n, m}=
\begin{pmatrix}
\frac{|f_{n,m-1}|^2}{|f_{n,m-1}|^2+|f_{n-1,m}|^2} & \frac{f^*_{n-1,m}f_{n,m-1}}{|f_{n,m-1}|^2+|f_{n-1,m}|^2} g \\
\frac{f_{n-1,m}f^*_{n,m-1}}{|f_{n,m-1}|^2+|f_{n-1,m}|^2}g^* & \frac{|f_{n-1,m}|^2}{|f_{n,m-1}|^2+|f_{n-1,m}|^2}
\end{pmatrix}\notag 
\end{align}
Similar to the calculation in section~\ref{appendixA}, one could show that the symmetric logarithmic derivative of  $\tilde{\sigma}^{AB\epsilon}_{n,m}$ in estimating $|g|$ and $\theta$ is given by:
\begin{widetext}
\begin{align*}
\hat{L}_{|g|}[\tilde{\sigma}^{AB\epsilon}_{n,m}]&=\frac{1}{1-|g|^2}\frac{2}{|f_{n,m-1}|^2+|f_{n-1,m}|^2}\begin{pmatrix} -|f_{n-1,m}|^2|g| &f_{n-1,m}f^*_{n,m-1}e^{i\theta}\\
f^*_{n-1,m}f_{n,m-1}e^{-i\theta}& -|f_{n,m-1}|^2|g| \end{pmatrix} \notag \\
\hat{L}_{\theta}[\tilde{\sigma}^{AB\epsilon}_{n,m}]&=|g|\frac{2}{|f_{n,m-1}|^2+|f_{n-1,m}|^2}\begin{pmatrix} 0&if_{n-1,m}f^*_{n,m-1}e^{i\theta}\\
-if^*_{n-1,m}f_{n,m-1}e^{-i\theta}& 0 \end{pmatrix}
\end{align*}
Therefore, we have:
\begin{subequations}
\begin{align}
 \mbb{H}_{|g|}[\sigma_{n,m}^{AB\epsilon}]&=\frac{(|f_{n,m-1}|^2+|f_{n-1,m}|^2)\epsilon}{2}\mbb{H}_{|g|}[\tilde{\sigma}^{AB\epsilon}_{n,m}]=\frac{\epsilon}{1-|g|^2}\frac{2|f_{n,m-1}|^2|f_{n-1,m}|^2}{|f_{n,m-1}|^2+|f_{n-1,m}|^2} \\
 \mbb{H}_{\theta}[\sigma_{n,m}^{AB\epsilon}]&=\frac{(|f_{n,m-1}|^2+|f_{n-1,m}|^2)\epsilon}{2}\mbb{H}_{\theta}[\tilde{\sigma}^{AB\epsilon}_{n,m}]=\epsilon|g|^2\frac{2|f_{n,m-1}|^2|f_{n-1,m}|^2}{|f_{n,m-1}|^2+|f_{n-1,m}|^2}
\end{align}
\label{eq:QFI-epsilon-app}
\end{subequations}
\end{widetext}
Combining Eq.~\ref{eq:QFI-app} and Eq.~\ref{eq:QFI-epsilon-app}, we conclude 
that to the leading-order term:
\begin{align}
&\mbb{H}_{|g|}[\mc{E}_{\text{g-ssr}}(\rho^{AB})]=\frac{\epsilon}{1-|g|^2}\sum_{n,m=1}\frac{2|f_{n,m-1}|^2|f_{n-1,m}|^2}{|f_{n,m-1}|^2+|f_{n-1,m}|^2}  \notag  \\
&\mbb{H}_{\theta}[\mc{E}_{\text{g-ssr}}(\rho^{AB})]=\epsilon|g|^2\sum_{n,m=1}\frac{2|f_{n,m-1}|^2|f_{n-1,m}|^2}{|f_{n,m-1}|^2+|f_{n-1,m}|^2}  \notag  \\
&\mathbb{h}[\rho^{AB}_a]=\sum_{n,m=1}\frac{2|f_{n,m-1}|^2|f_{n-1,m}|^2}{|f_{n,m-1}|^2+|f_{n-1,m}|^2}    
\end{align}
\label{lem:QFInm}
where the sum is over $n,m\ge 1$, since the contribution from terms with $n=0$ or $m=0$ are apparently zero. 
\end{proof}
\blk

Let us now discuss the shared ancilla state $\ket{\psi}_a$ with bounded energy. A natural way to implement this restriction is to limit the maximum total photon number to a fixed value $N$, which yields the following theorem. 
\begin{theorem}
For ancilla state $\ket{\psi}_a$ with at most $N$-photon in total, we have:
\begin{align}
\mathbb{h}[\rho^{AB}_a]\le \cos(\frac{\pi}{N+2})\stackrel{\text{Large $N$}}{\approx} 1-\frac{\pi^2}{2N^2}
\end{align}
\label{thm: QFI-upper}
\end{theorem}
\begin{proof}
Assuming the ancilla state uses at most $N$ photons, we need only consider the subnormalized blocks $\sigma_{n,m}^{AB}$ with $n+m\le N+1$ (with at most $N$ photons from the ancilla state and one photon from the source). 
It is thus convenient to rewrite the summation as:
\begin{align}
    \mathbb{h}[\rho^{AB}_a]&=\sum_{n,m=1}^{n+m\le N+1}\frac{2|f_{n,m-1}|^2|f_{n-1,m}|^2}{|f_{n,m-1}|^2+|f_{n-1,m}|^2}\notag \\   
&=\sum_{k=1}^{N}\sum_{n=1}^{k}\frac{2|f_{n,k-n}|^2|f_{n-1,k-n+1}|^2}{|f_{n,k-n}|^2+|f_{n-1,k-n+1}|^2}
\end{align}
For fixed $k$, define:
$$q_k=\sum_{n=0}^{k}|f_{n,k-n}|^2\quad\quad\quad {p_{n|k}}:=\frac{|f_{n,k-n}|^2}{q_k}$$

\blk
where $\sum_k q_k=1$ and $\sum_n p_{n|k}=1$ representing the probability of having $k$ total photon from the ancilla state $\rho^{AB}_a$ and conditional on that, the probability having $n$ photon at telescope A given $k$. With this notation, we could show that:
\begin{align}
&\mathbb{h}[\rho^{AB}_a]=\sum_{k=1}^N h_k \\
&h_k:=q_k \sum_{n=1}^{k}\frac{2p_{n|k}p_{n-1|k}}{p_{n|k}+p_{n-1|k}}\le q_k \sum_{n=1}^{k}{\sqrt{p_{n|k}p_{n-1|k}}}
   \label{eq:upper bound}
\end{align}
Now, we can give a non-tight upper bound $h_k$ by upper bounding $\sum_{n=1}^{k}\sqrt{p_{n|k}p_{n-1|k}}$ from the lemma~\ref{lemma4} below, which shows 
\begin{align}
  &h_k\le q_k\cos(\frac{\pi}{k+2}) \\
  &\mathbb{h}[\rho^{AB}_a]=\sum_{k=1}^{N} h_k\le \sum_{k=1}^{N} q_k\cos(\frac{\pi}{k+2})\le \cos(\frac{\pi}{N+2})\notag
\end{align}
with the second inequality saturated when $q_k=\delta_{k,N}$, which implies:
$$f_{n,m}=f_{n,m}\delta_{n,N-m}=:f_{n}\delta_{n,N-m}$$
Therefore, in the following, we focus on states of the form
\begin{equation}
   \ket{\psi}_a=\sum_{n=0}^Nf_{n}\ket{n_A,(N-n)_B}, 
\label{eq: ancilla-state-optimal}
\end{equation}
\end{proof}
\begin{lemma}
\label{lemma4}
For any distribution $\{x_i\}$ with $\sum_{i=0}^kx_i=1$, we have the inequality:
\begin{equation}
   \sum_{i=1}^k \sqrt{x_ix_{i-1}}\le\cos(\frac{\pi}{k+2}),
\end{equation}
with equality for $x_i={\frac{2}{k+1}}\sin^2(\frac{i+1}{k+2}\pi)$
\end{lemma}
\label{lem: Cheb}
\begin{proof}
Defining a $(k+1)\times(k+1)$  tridiagonal matrix $M_{k+1}$ of the form:
\begin{equation}
\begin{pmatrix}
0 & 1 & 0 & \cdots & 0 \\
1 & 0 & 1 & \cdots & 0 \\
0 & 1 & 0 & \cdots & 0 \\
\vdots & \vdots & \ddots & \ddots & 1 \\
0 & 0 & \cdots & 1 & 0
\end{pmatrix}.
\end{equation}
One could easily see that:
\begin{align}
2\sum_{i=1}^k \sqrt{x_i x_{i-1}}=\vec v^{T} M_{k+1} \vec v.
\end{align}
with normalized vector $\vec{v}=(\sqrt{x_0},\cdots,\sqrt{x_k})^T$.\par 

Therefore, maximizing $2\sum_{i=1}^k \sqrt{x_ix_{i-1}}$ for all distribution $\{x_i\}$ is equivalent to compute the largest eigenvalue of $M_{k+1}$ (see also similar quantity discussed in \cite{wiseman1997, Bartlett2006b}), whose characteristic polynomials satisfy
$$
\det(M_{k+1} - \lambda I_{k+1}) = -\lambda \cdot \det(M_{k} - \lambda I_k) - \det(M_{k-1} - \lambda I_{k-1})$$
where $I_k$ are $k\times k$ identity matrix. This recursive formula $\det(M_k - \lambda I_k)$ matches that of the Chebyshev polynomials $U_{k+1}(-\lambda/2)$, whose eigenvalues are 
$$\lambda_i=2\cos(\frac{i+1}{k+2}\pi)\quad\quad\quad i\in[0,\cdots k]$$
with the largest eigenvalue $2\cos(\frac{\pi}{k+2})$ and eigenvector $\vec{v}=\sqrt{\frac{2}{k+1}}(\sin(\frac{1}{k+2}\pi),\sin(\frac{2}{k+2}\pi),\cdots,\sin(\frac{k+1}{k+2}\pi))^T$. 

Therefore, we conclude that $$\sum_{i=1}^k \sqrt{x_ix_{i-1}}\le\cos(\frac{\pi}{k+2}),$$ with equality for $x_i={\frac{2}{k+1}}\sin^2(\frac{i+1}{k+2}\pi)$.
\end{proof}
\begin{lemma}\label{lem:harmonic}
For any distribution $\{x_i\}$ with $\sum_{i=0}^kx_i=1$ that is sufficiently smooth (defined below). Define $S_k(x) :=\;\sum_{i=1}^{k}\frac{2\,x_i\,x_{i-1}}{x_i+x_{i-1}}$, for large $k$, one has:
\begin{align} 
\sup_{\{x_i\}}\;S_k(x)
		&=\;1-\frac{\pi^{2}}{k^{2}}+O\!\bigl(\frac{1}{k^3}\bigr)
\end{align}
\end{lemma}
\begin{proof}
We first define $f_i=\sqrt{x_i}$ with normalization $\sum f_i^{2}=1$.
Expanding the denominator to leading order gives
\begin{equation}
	S_k(x)=1-\frac{f_0^{2}+f_k^{2}}{2}-\sum_{i=1}^{k}\left[(f_i-f_{i-1})^{2}+O\!\big((f_i-f_{i-1})^{4}\big)\right].
\end{equation}

Assume $f_i$ is sufficiently smooth so that we can write $s=i/k$ and set $f_i=k^{-1/2}u(s)$, where $u'(x)$ is Lipschitz (i.e., $|u'(x)-u'(y)|\le C|x-y|$). For large $k$, we then have 
\blk
\begin{align}
  &\sum_{i=1}^{k}\left[(f_i-f_{i-1})^{2}+O\!\big((f_i-f_{i-1})^{4}\big)\right]\notag  \\
   =&\frac{1}{k^2}\!\int_{0}^{1}\!u'(s)^{2}\,ds+O\!\Big(\frac{1}{k^3}\Big),
\end{align}
hence
\begin{equation}
	S_k(x)=1-\frac{1}{k^{2}}\!\int_{0}^{1}\!u'(s)^{2}\,ds
	         -\frac{u(0)^{2}+u(1)^{2}}{2k}+O\!\Big(\frac{1}{k^3}\Big). \label{eq:continuous}
\end{equation}
Placing any positive mass at the ends ($u(0),u(1)>0$) yields a loss of order $1/k$, so the maximiser must satisfy $u(0)=u(1)=0$. The problem then reduces to minimising $\int_0^1 u'(s)^{2}ds$ subject to $\int_0^1 u(s)^{2}ds=1$ and $u(0)=u(1)=0$, i.e. the Dirichlet eigenvalue problem
\[
u''+\lambda u=0,\qquad u(0)=u(1)=0,
\]
whose lowest eigenpair is $\lambda_1=\pi^{2}$ with $u_1(s)=\sqrt{2}\sin(\pi s)$. Plugging $\lambda_1=\pi^{2}$ into \eqref{eq:continuous} gives
\[
S_k^{\max}=1-\frac{\pi^{2}}{k^{2}}+O\!\Big(\frac{1}{k^3}\Big).
\]
\end{proof}
Plugging this into Eq.~\ref{eq:upper bound}, we get an better upper bound of $\mbb{h}[\rho_a^{AB}]$ in the asymptotical region:
\begin{equation}
  \mbb{h}[\rho_a^{AB}]  \stackrel{\text{Large $N$}}{\approx} 1-\frac{\pi^2}{N^2}+O(\frac{1}{N^3})
\end{equation}

Another way to impose the constraint could be on the average photon number $\langle N\rangle=\sum_{n,m}|f_{n,m}|^2(n+m)$, where a similar bound could be given as:
\begin{corollary}
For ancilla state $\ket{\psi}_a$ with average photon number $\langle N\rangle$, we have:
\begin{align}
\mathbb{h}[\rho^{AB}_a]\le \cos(\frac{\pi}{\langle N\rangle+2})\stackrel{\text{Large $\langle N\rangle$}}{\approx} (1-\frac{\pi^2}{2\langle N\rangle^2})
\end{align}    
\end{corollary}
\begin{proof}
Defining $q_k=\sum_{n=0}^k|f_{n,k-n}|^2$, the average photon number $\langle N\rangle=\sum_{n,m}|f_{n,m}|^2(n+m)$ can be rewritten as:
$$\langle N\rangle=\sum_{k=0}^{\infty}kq_k $$
And we now solve the following optimization problem 
\begin{align*}
&\text{sup}\sum_{k=0}^{\infty} q_k\cos(\frac{\pi}{k+2}) \\
&\text{s.t. }\sum_{k=0}^{\infty} q_k=1 \text{ and }\sum_{k=0} q_k k=\langle N \rangle
\end{align*}
Since $\cos(\frac{\pi}{k+2})$ is a convex function, from Jensen's inequality ($E(\Phi(x))\le \Phi(E(X))$ for any random variable $X$ and convex function $\Phi$):
$$\sum_{k=0}^{\infty} q_k\cos(\frac{\pi}{k+2})\le \cos(\frac{\pi}{\langle N\rangle+2}) $$
And the inequality saturate if $\langle N\rangle $ is an integer and $q_k= \delta_{k,\langle N\rangle}$
\end{proof}

\section{Discussion of different entanglement-assisted protocols}
\subsection{Single-particle entangled states}
\label{appendixD1}
We begin this section by discussing the Gottesman–Jennewein–Croke (GJC) scheme \cite{Gottesman2012} within our theoretical framework. In this scheme, a single-particle entangled state, $\ket{\psi}_a = \frac{1}{\sqrt{2}}(\ket{0}_A\ket{1}_B + \ket{1}_A\ket{0}_B)$, is used to implement a quantum telescope. Applying Proposition~\ref{prop: assisted-QFI}, it is straightforward to verify that:
\begin{equation}
    \mathbb{h}[\ket{\psi}_a]=\frac{1}{2} \notag
\end{equation}
The QIF ratio bound establishes a fundamental limit for quantum estimation for quantum telescopes employing single-particle entangled states $\ket{\psi}_a$, and it has already been achieved via local, correlated interference in the original quantum telescopy scheme \cite{Gottesman2012, tsang2011}. While \cite{Czupryniak2023} argues that certain nonlinear-optical measurements can increase the classical Fisher information (FI) for particular setups, this does not contradict our result: our quantum Fisher information (QFI) upper bound still applies. The apparent tension disappears once we note that the coherent-control gates proposed in \cite{Czupryniak2023} would have to act at two spatially separated telescopes and are not implementable under a $U(1)$ superselection rule without additional resources that circumvent the SSR. Since under a 
$U(1)$ SSR, both states and allowed operations must be $U(1)$-invariant \cite{Barlett2006}. Consequently, such gates cannot be implemented at both sites when their phase references are uncorrelated.

Although the original GJC scheme cannot be improved by changing only the local operations or measurements, one can enhance performance by changing the ancilla. Next, we analyze a natural extension that shares $N$-copy single-particle entangled states between the telescopes.
\subsection{N-copy single-particle entangled states}
\label{appendixDA}
We note that the use of \textit{N}-copy SPE, $\ket{\psi}_a^{\otimes N} = \frac{1}{\sqrt{2^N}}(\ket{0}_A\ket{1}_B + \ket{1}_A\ket{0}_B)^{\otimes N}$ has been considered previously in \cite{Czupryniak2022, Marchese2023}. However, their approach does not fully resolve the general problem, as we will discuss later. More importantly, the performance of N-copy SPE has not yet been thoroughly analyzed. Applying Proposition~\ref{prop: assisted-QFI}, for the N-copy SPE we obtain:
\begin{equation}
\ket{\psi}^{\otimes N}_a=\sum_{n=0}^{N}\sqrt{\frac{{N\choose n}}{2^N}}\ket{n_A,(N-n)_B} \Rightarrow     \mathbb{h}[\ket{\psi}^{\otimes N}_a]=\frac{N}{N+1},   \notag
\end{equation}
Note that the state $\ket{n_A,(N-n)_B}$ is not a product state, but instead mode-entangled of the form:
\begin{equation}
  \ket{n_A,(N-n)_B}=\frac{1}{\sqrt{{N\choose n}}}\sum_{\pi}\pi(\ket{1}_A^{\otimes n} \ket{0}_A^{\otimes m}\otimes \ket{0}_B^{\otimes n}\ket{1}_B^{\otimes m})
\end{equation}
where $\pi$ represents different joint local permutations on the photon modes (There are $N=n+m$ modes in both telescope A and B, so there are in total ${N\choose n}$ such configurations. ).
\par 

The QFI ratio approaches one as the number of SPE copies increases. In certain special cases, this bound has been saturated using various local measurement schemes, including the estimation of a point source (source with $|g|-1$) with linear optics \cite{Marchese2023}, and the estimation of a general source using nonlinear operations \cite{Czupryniak2022}. Here, we provide brief comments on these two results.

In \cite{Czupryniak2022}, a phase extraction scheme was proposed in Section IV, where it was argued that their scheme could be done with only linear-optical elements. However, we found out that the operation required to transform Eq. 51 to Eq. 52 in their paper cannot be implemented with linear optics. More specifically, linear-optical operations are not capable of performing quantum Fourier transformations over different $N$ photon states for $N>1$ (i.e., $\hat{a}_2\hat{a}_3 \nRightarrow \frac{1}{\sqrt{3}}(\hat{a}_2\hat{a}_3+\hat{a}_1\hat{a}_3+\hat{a}_1\hat{a}_2)$ is not allowed), since the linear optics operations should not only be linear but also be described as a unitary operation on individual creation operator.(i.e., $\hat{a}_3 \Rightarrow \frac{1}{\sqrt{3}}(\hat{a}_3+\hat{a}_1+\hat{a}_2)$ is allowed). However, if one allows LOCC beyond passive linear optics, the phase-extraction in \cite{Czupryniak2022} can estimate both $\theta$ and $|g|$, saturating the Quantum Fisher information ratio we report here. 
\par 

In \cite{Marchese2023}, a protocol is presented that employs only linear optics to estimate the phase $\theta$ of a point source (with $|g|=1$). This protocol achieves a Fisher information ratio for estimating $\mathbb{f}^N_{\theta}$ that scales as $1 - \frac{1}{N+1}$, where $N$ is the number of copies of the shared single-particle entangled state (SPE). In this work, we extend their results by explicitly computing the Fisher information for estimating both $|g|$ and $\theta$ using their scheme for $N \in \{1, \ldots, 5\}$.
\begin{align}
&\mathbb{F}_{\theta}^1=\frac{1}{2}\frac{\sin^2\phi}{1-|g|^2\cos^2\phi}|g|^2\epsilon\notag \\ &\mathbb{F}_{|g|}^1=\frac{1}{2}\frac{\cos^2\phi}{1-|g|^2\cos^2\phi}\epsilon \notag\\
&\mathbb{F}_{\theta}^2=\frac{3\sin^2\phi}{(1-|g|\cos\phi)(5+4|g|\cos\phi)}|g|^2\epsilon
\notag \\
& \mathbb{F}_{|g|}^2=\frac{3\cos^2\phi}{(1-|g|\cos\phi)(5+4|g|\cos\phi)}\epsilon\notag\\
&\mathbb{F}_{\theta}^3=\frac{3\sin^2\phi(9+7|g|\cos\phi)}{4(1-|g|^2\cos^2\phi)(10+6|g|\cos\phi)}|g|^2\epsilon\notag \\
&\mathbb{F}_{|g|}^3=\frac{3\cos^2\phi(9+7|g|\cos\phi)}{4(1-|g|^2\cos^2\phi)(10+6|g|\cos\phi)}\epsilon \notag\\
&\mathbb{F}_{\theta}^4=\frac{10\sin^2\phi(16+9|g|\cos\phi)}{(1-|g|\cos\phi)(13+12|g|\cos\phi)(17+8|g|\cos\phi)}|g|^2\epsilon\notag \\
&\mathbb{F}_{|g|}^4=\frac{10\cos^2\phi(16+9|g|\cos\phi)}{(1-|g|\cos\phi)(13+12|g|\cos\phi)(17+8|g|\cos\phi)}\epsilon \notag\\
&\mathbb{F}_{\theta}^5=\frac{5\sin^2\phi(79+106|g|\cos\phi+31|g|^2\cos^2\phi)}{(1-|g|^2\cos^2\phi)(26+10|g|\cos\phi)(20+16|g|\cos\phi)}|g|^2\epsilon\notag \\
&\mathbb{F}_{|g|}^5=\frac{5\cos^2\phi(79+106|g|\cos\phi+31|g|^2\cos^2\phi)}{(1-|g|^2\cos^2\phi)(26+10|g|\cos\phi)(20+16|g|\cos\phi)}\epsilon
\label{eq: FI-point}
\end{align}
where $\phi=\delta-\theta$ with $\delta$ being the phase delay used in the measurement. In \cite{Marchese2023}, they showed implicitly that $\phi=0$ is the optimal point to get the optimal Fisher information ratio $\mathbb{f}^N_{\theta}=1-\frac{1}{N+1}$ for a point source with $|g|=1$, which agrees with Eq.~\ref{eq: FI-point} above. However, one could easily check that $\mathbb{f}^N_{\theta}=0$ at $\phi=0$ for a source with $|g|<1$ (which is almost always the case in astronomical interferometric imaging), in fact the Fisher information ratio scaling $1-\frac{1}{N+1}$ is an artifact that occurs because $|g|=1, \phi=0$ is a removable singular point. As some concrete examples, one can verify that:
\begin{align}
\mathbb{F}_{\theta}^2&=\frac{3\sin^2\phi}{(1-|g|\cos\phi)(5+4|g|\cos\phi)}|g|^2\epsilon \notag \\
&\stackrel{\text{$|g|=1$, $\phi\rightarrow 0$}}{\Longrightarrow} \frac{3\phi^2}{(\frac{1}{2}\phi^2)(5+4)}\epsilon =\frac{2}{3}|g|^2\epsilon  \\
\mathbb{F}_{\theta}^2&=\frac{3\sin^2\phi}{(1-|g|\cos\phi)(5+4|g|\cos\phi)}|g|^2\epsilon \notag \\
&\stackrel{\text{$|g|<1$, $\phi\rightarrow 0$}}{\Longrightarrow} 0\\
\mathbb{F}_{\theta}^2&=\frac{3\sin^2\phi}{(1-|g|\cos\phi)(5+4|g|\cos\phi)}|g|^2\epsilon \notag \\
&\stackrel{\text{$|g|\rightarrow 0$, $\phi\rightarrow \frac{\pi}{2}$}}{\Longrightarrow} \frac{3}{5}|g|^2\epsilon 
\end{align}
Unfortunately, we don't have a close-form answer on the achievable Fisher information ratio $\mathbb{f}^N_{\theta}$ and $\mathbb{f}^N_{|g|}$ using their scheme\cite{Marchese2023}, because the computational complexity of their protocols (And in fact, any protocols that involves $N$-copy SPE states and linear-optical circuit $\mc{U}$) is essentially related to the computational complexity of Boson sampling\cite{Ticky2010}.  However, from our observations, we conjecture that:
\begin{conjecture}
linear-optical implementation of $N$-copy SPE achieves optimal quantum telescopy asymptotically, however, the FI ratio $\mathbb{f}_{\mu}^N$ does not saturate the QFI ratio  $\mathbb{h}^N$ and scales $\mathbb{f}_{\mu}^N \approx 1-\frac{2}{N+3}\stackrel{\text{Large $N$}}{\approx} 1-\frac{2}{N}$ for $\mu\in [|g|, \theta]$
\end{conjecture}
One could plug in $\phi=0$ for $\mathbb{f}_{|g|}^N$ and $\phi=\pi/2$ for $\mathbb{f}_{\theta}^N$ to Eq.~\ref{eq: FI-point} to check our conjecture. 
\subsection{KLM and modified KLM schemes}
\label{appendixDB}
As established in the previous section, an optimal ancilla state (for a given total photon budget) has the form:
$\ket{\psi}_a=\sum_{n=0}^Nf_{n}\ket{n_A,(N-n)_B}$. 
Accordingly, from now on, we consider shared ancilla states
\begin{align}
\ket{\psi}_a&=\sum_{n=0}^Nf_{n}\ket{n_A,(N-n)_B} \label{eq:LOQC source}\\
&=\sum_{n=0}^Nf_{n}\ket{1}^{\otimes n}_A\ket{0}^{\otimes N-n}_A\ket{0}^{\otimes n}_B\ket{1}^{\otimes N-n}_B.\notag 
\end{align}
And the QFI ratio in proposition~\ref{prop: assisted-QFI} gets simplified to:
\begin{equation}
 \mathbb{h}[\rho^{AB}_a]=\sum_{n=1}^N\frac{2|f_{n}|^2|f_{n-1}|^2}{|f_{n}|^2+|f_{n-1}|^2}     
\end{equation}
\blk

In this subsection, we discuss schemes inspired by teleportation tricks used in linear-optical quantum computation.~\cite{Knill2001, Franson2002}. In the original paper by Knill, a $2N$-mode state ($N$ modes on Alice and $N$ modes on Bob\blk) of the following form is prepared
$$\ket{\psi}^{\text{KLM}}_a=\frac{1}{\sqrt{N+1}}\sum\limits_{n=0}^N\ket{1}^{\otimes n}_A\ket{0}^{\otimes N-n}_A\ket{0}^{\otimes n}_B\ket{1}^{\otimes N-n}_B.$$
We will provide a detailed discussion of the implementation of state teleportation using shared entangled states of the above form in a later section. Briefly, in the KLM teleportation protocol, teleportation succeeds perfectly except with failure probability $\frac{1}{N+1}$, which occurs when no photon arrives at either telescope A or B. Applying Proposition~\ref{prop: assisted-QFI} to this scenario, we obtain: 
\begin{equation}
\mathbb{h}[\ket{\psi}_a^{\text{KLM}}]=\frac{N}{N+1},     
\end{equation}
which scales the same as the $N$-copy SPE; however, as we will show in a later section, this bound is achievable using LOCC and linear optics. 
\par 
A natural question is whether other shared entangled states (i.e., different choices of amplitudes $f_n$)can yield improved teleportation and hence a larger QFI ratio. Related questions have been studied previously, with performance quantified by average teleportation fidelity (Franson 2002). The intuition behind these protocols is to reduce the failure probability of the teleportation protocol (to be less than $\frac{1}{N+1}$) by allowing minor errors upon successful events.

It is important to note that failure events occur precisely when no photon reaches either telescope A or B. Thus, a natural amplitude-shaping strategy is to suppress $f_0$ and $f_{N}$. Motivated by this, we analyze two modified KLM families: one with triangular intensity $|f_n|^2$ and one with triangular amplitude $|f_n|$

In the "triangular intensity" modification, we take $$\ket{\psi}^{\text{Itri}}_a=\sum\limits_{n=0}^Nf^{\text{Itri}}_n\ket{1}^{\otimes n}_A\ket{0}^{\otimes N-n}_A\ket{0}^{\otimes n}_B\ket{1}^{\otimes N-n}_B$$
with $|f^{\text{Itri}}_n|^2\propto \frac{N}{2}-|\frac{N}{2}-n|$, i.e., the intensity of different terms shaped as a triangle. We assume $N$ as an even number for simplicity (the odd number case follows similarly), we have $|f^{\text{Itri}}_n|^2=\frac{4}{N^2}(\frac{N}{2}-|\frac{N}{2}-n|)$, which gives:
\begin{align}
&\mathbb{h}[\ket{\psi}^{\text{Itri}}_a]=\frac{16}{N^2}\sum_{n=1}^{N/2}\frac{i(i-1)}{i+i-1}=\frac{16}{N^2}\sum_{n=1}^{N/2}(\frac{i}{2}-\frac{1}{4}-\frac{1}{4}\frac{1}{2i-1}) \notag \\
&=1-\frac{4}{N^2} \sum_{n=1}^{N/2}\frac{1}{2i-1}\stackrel{\text{Large $N$}}{\approx} 1-2\frac{\log(N)}{N^2},  
\end{align}
where we used the fact that $\sum_{n=1}^{N/2}\frac{1}{2i-1}\approx \frac{1}{2}(\log(2N)+\gamma)$. \\

Similarly, in the "triangular amplitude," we could modify the amplitude as
$$\ket{\psi}^{\text{Atri}}_a=\sum\limits_{n=0}^Nf^{\text{Atri}}_n\ket{1}^{\otimes n}_A\ket{0}^{\otimes N-n}_A\ket{0}^{\otimes n}_B\ket{1}^{\otimes N-n}_B$$
where $|f^{\text{Atri}}_n|\propto \frac{N}{2}-|\frac{N}{2}-n|$, i.e., the amplitude of different terms shaped as a triangle. Where again, take $N$ to be an even number for simplicity, we have $|f^{\text{Atri}}_n|^2=\frac{12}{N(N^2+2)}(\frac{N}{2}-|\frac{N}{2}-n|)^2$, which gives:
\begin{align}
&\mathbb{h}[\ket{\psi}^{\text{Atri}}_a]=\frac{48}{N(N^2+2)}\sum_{n=1}^{N/2}\frac{i^2(i-1)^2}{i^2+(i-1)^2}
\notag \\
&=\frac{48}{N(N^2+2)}\sum_{n=1}^{N/2}(\frac{i^2}{2}-\frac{i}{2}-\frac{1}{4}+\frac{1}{4}\frac{1}{2i^2-2i+1})\notag \\
&\stackrel{\text{Large $N$}}{\approx} 1-12\frac{1}{N^2},  
\end{align}
This achieves a better asymptotic QFI ratio and is worse than the analytic upper bound only by a constant factor. However, as shown in Fig. 2, for small $N$, the intensity-triangle (first modified KLM) performs better.

\subsection{Optimal-KLM scheme}
\label{appendixDC}
We now present the best scheme we identified; numerically, it is very close to the optimum. The shared entangled state considered here takes the form
\begin{equation}
    \ket{\psi}^{\text{Opt}}_a=\sqrt{\frac{2}{N+2}}\sum\limits_{n=0}^N\sin(\frac{n+1}{N+2}\pi)\ket{1}^{\otimes n}_A\ket{0}^{\otimes N-n}_A\ket{0}^{\otimes n}_B\ket{1}^{\otimes N-n}_B, 
\end{equation}
which is inspired by the state used to obtain the analytical upper bound and has appeared in several other quantum information tasks \cite{Barlett2007, wiseman1997, Barlett2006}. Recall that in Lemma~\ref{lem: Cheb}, we showed that a state vector of the above form maximizes the quantity $\sum_{n=1}^N |f_n f_{n-1}|$. However, the quantity of interest here is the quantum Fisher information ratio:
$$ \mathbb{h}[\rho^{AB}_a]=\sum_{n=1}^N\frac{2|f_{n}|^2|f_{n-1}|^2}{|f_{n}|^2+|f_{n-1}|^2}  $$
Plugging $|f_{n}|=\sqrt{\frac{2}{N+2}}\sin(\frac{n+1}{N+2}\pi)$ in, we have:
\begin{align}
&\mathbb{h}[\rho^{AB}_a]=\frac{4}{N+2}\sum_{n=1}^N\frac{\sin^2(\frac{n+1}{N+2}\pi)\sin^2(\frac{n}{N+2}\pi)}{\sin^2(\frac{n+1}{N+2}\pi)+\sin^2(\frac{n}{N+2}\pi)}\notag \\
&=\frac{1}{N+2}\sum_{n=1}^N\frac{[\cos(\frac{\pi}{N+2})-\cos(\frac{2n+1}{N+2}\pi)]^2}{1-\cos(\frac{\pi}{N+2})\cos(\frac{2n+1}{N+2}\pi)} \notag\\
&= \frac{\cos(\frac{\pi}{N+2})}{N+2}\sum_{n=1}^N[\cos(\frac{\pi}{N+2})-\cos(\frac{2n+1}{N+2}\pi)]\notag \\
&-\frac{\sin^2(\frac{\pi}{N+2})}{N+2}\sum_{n=1}^N\frac{\cos(\frac{2n+1}{N+2}\pi)[\cos(\frac{1}{N+2}\pi)-\cos(\frac{2n+1}{N+2}\pi)]}{1-\cos(\frac{\pi}{N+2})\cos(\frac{2n+1}{N+2}\pi)}
\end{align}
where the first term can be simplified, using the identity $\sum_{n=0}^{N+1}\cos(\frac{2n+1}{N+2}\pi)=0$, thus
$$\frac{\cos(\frac{\pi}{N+2})}{N+2}\sum_{n=1}^N[\cos(\frac{\pi}{N+2})-\cos(\frac{2n+1}{N+2}\pi)]=\cos^2(\frac{\pi}{N+2}).$$
Moreover, the coefficient $\frac{\sin^2\left(\frac{\pi}{N+2}\right)}{N+2} = O\left(\frac{1}{N^3}\right)$ in the second term, while the remaining factor $\sum_{n=1}^N\frac{\cos(\frac{2n+1}{N+2}\pi)[\cos(\frac{1}{N+2}\pi)-\cos(\frac{2n+1}{N+2}\pi)]}{1-\cos(\frac{\pi}{N+2})\cos(\frac{2n+1}{N+2}\pi)}$ converge to a negative constant independent of $N$ (numerically, it converges to $-\frac{11\pi}{12}$ ). \par 
Thus, we conclude that:
\begin{equation}
\mathbb{h}[\rho^{AB}_a]\approx \cos^2(\frac{\pi}{N+2}) \stackrel{\text{Large $\langle N\rangle$}}{\approx} 1-\frac{\pi^2}{N^2}
\end{equation}
\subsection{Two-mode squeezed vacuum state}
\label{appendixDD}
The two-mode squeezed vacuum state (TMSV) has served as the standard bipartite entangled state in continuous-variable systems \cite{Weedbrook2012, Braunstein2005}. It was pointed out in the original quantum telescopy paper by Gottesman \cite{Gottesman2012} and later in several other works \cite{wang2023astronomical, huang2024} that continuous-variable quantum teleportation \cite{Vaidman1994, Braunstein1998} using TMSV can facilitate quantum telescopy. \par 
A TMSV state is written as 
\begin{equation} \ket{\psi}^{\text{CV}}_a=\frac{1}{\cosh(r)}\sum\limits_{n=0}^N \tanh^n{r}\ket{n}_A\ket{n}_B\end{equation}\\
where the squeezing parameter $r$ is related to the total average photon number $\langle N\rangle=2\sinh^2(r)$, so equivalently, one could also express it as:
\begin{equation}\ket{\psi}^{\text{CV}}_a=\sum\limits_{n=0}^{\infty} \frac{\sqrt{2\langle N\rangle^{n}}}{\sqrt{(2+\langle N\rangle)^{n+1}}}\ket{n}_A\ket{n}_B]\end{equation}
At first glance, proposition~\ref{prop: assisted-QFI} implies that $$\mathbb{h}[\ket{\psi}^{\text{CV}}_a] = 0,$$ 
which seemingly contradicts the CV protocol introduced in \cite{wang2023astronomical, huang2024}, but aligns with the recent argument in \cite{purvis2024}, where TMSV is utilized, but CV-teleportation is not implemented.

However, this apparent contradiction can be resolved by closely examining the standard continuous-variable quantum teleportation protocol \cite{Braunstein1998}, which involves three key steps:
\begin{itemize} \item[(1)] Mixing the state (to be teleported) with one mode of the TMSV on a beam splitter. \item[(2)] Applying a homodyne measurement to the joint state after mixing. \item[(3)] Sending the measurement outcomes of the homodyne measurement $(x_1, p_2)$ (two quadratures at two outputs) is sent to the other mode of TWSQ, where a displacement operator $\mc{D}(x_1,p_2)$ is implemented correspondingly. \end{itemize} 
Crucially, steps (2)–(3) require a shared phase reference so that both homodynes and the displacement are phase-aligned. As emphasized in \cite{Furusama1998, Rudolph2001}. More precisely, an extra pair of correlated coherent state $\ket{\alpha}_A\ket{\alpha}_B$ needs to be shared, such that Homodyne measurement at step (2) and displacement operator at step (3) can be done (note that this is different from the classical information send from A to B, but an extra pre-shared correlated coherent state (or phase-locked synchronized coherent state) to ensure the Homodyne measurement and displacement operator are performed with the same phase reference).\par 
Including this extra classical correlated coherent state, the overall ancilla state used in the CV teleportation protocol is:
\begin{equation}
    \ket{\psi}^{\text{CV}}_a=\frac{1}{\cosh(r)}\sum\limits_{n=0}^N \tanh^n{r}\ket{n}_A\ket{n}_B\ket{\alpha}_A\ket{\alpha}_B,
\end{equation}
which admits a non-vanishing QFI ratio according to Proposition~\ref{prop: assisted-QFI}.

To avoid the messy calculation (ideal CV teleportation corresponds to $|\alpha|\rightarrow \infty$), here we use results derived in~\cite{wang2023astronomical} to compare its performance. It was shown in~\cite{wang2023astronomical} that: 
\begin{widetext}
\begin{align}
& \mathbb{f}_{\theta}=\frac{2\epsilon^2|g|^2}{2y+\epsilon(2+\epsilon-\epsilon|g|^2+2y)}\\
&\mathbb{f}_{|g|}=\frac{2\epsilon^2[-\epsilon(2+\epsilon)^2+\epsilon^3|g|^4-4(1+\epsilon)(2+\epsilon)y-4(2+\epsilon)y^2]}{[\epsilon(-1+|g|^2)-2y][\epsilon(-2-\epsilon+\epsilon|g|^2)-2(1+\epsilon)y][\epsilon^2(-1+|g|^2)-4(1+y)-2\epsilon(2+y)]}
\end{align}    
\end{widetext}
Where $y$ is the squeezing parameter $y=2e^{-2r}$. Replacing $y=2(\sqrt{\frac{\langle N\rangle}{2}}+\sqrt{\frac{\langle N\rangle}{2}+1})^{-2}$, as $\langle N\rangle\rightarrow \infty$ we have, $y=\frac{1}{N}$ and thus:
\begin{equation}
  \mathbb{h}[\ket{\psi}^{\text{CV}}_a]\ge \mathbb{f}[\ket{\psi}^{\text{CV}}_a] \stackrel{\text{ $N\gg \frac{1}{\epsilon}$}}{\approx} 1-\frac{1}{\epsilon \langle N\rangle}, 
\end{equation}
Which is worse than the optimal scheme we proposed in the above sections, even if we grant an unbounded correlated coherent state $\ket{\alpha}_A\ket{\alpha}_B$.

We also note that this FI ratio result above does not contradict the following subsections on correlated coherent state, where the QFI ratio for that scheme is discussed. Here, we present an explicit attainable Fisher information, which, in general, only serves as a lower bound on the QFI ratio. 
\blk

\subsection{Two particle entangled state, and NOON state}
\label{appendixDE}
In this section, we utilize our framework to address a common question in the community: Can a standard two-photon, four-mode entangled state aid in quantum telescopy? \par 
A two-photon entangled state can be defined as : 
\begin{align}
    \ket{\psi}^{\text{TPE}}_a &= \frac{1}{\sqrt{2}}(\ket{1}_{A_1}\ket{0}_{A_2}\ket{0}_{B_1}\ket{1}_{B_2} + \ket{0}_{A_1}\ket{1}_{A_2}\ket{1}_{B_1}\ket{0}_{B_2})\\
    &=\ket{1_A,1_B}\in \mc{H}_1^A\otimes \mc H^B_1 \notag, 
\end{align}
This state represents a form of dual-rail encoding in quantum information theory \cite{Kok2007} (while the single-particle entangled state is referred to as single-rail encoding). However, by Proposition~\ref{prop: assisted-QFI}, this state yields zero QFI ratio.

The reason is simple: the standard two-photon entangled state is not a phase-reference (PR) resource, so it cannot establish a shared phase reference. This distinction between nonlocal (entangled) resources and phase-reference resources was already emphasized in \cite{Schuch2004}.
\par 
A natural misconception is that discarding one local mode “activates” a PR resource. However, this is not the case. Specifically, tracing out mode $A_2$ and $B_2$ produces an \emph{incoherent mixture} on $A_1B_1$,
$\tr_{A_2B_2}\big[\op{\psi}{\psi}|^{\text{TPE}}_a\big]
=\frac{1}{2}\big(|10\rangle\langle10|_{A_1B_1} + |01\rangle\langle01|_{A_1B_1}\big),$
not the single-particle entangled state; consequently, the QFI remains zero, agreeing with the fact that QFI is non-increasing under postprocessing. 
\par 

Following the same logic, we briefly examine the NOON state:
\begin{equation}
\ket{\psi}^{\text{NOON}}=\frac{1}{\sqrt{2}}\ket{N}_A\ket{0}_B+\ket{0}_A\ket{N}_B.
\end{equation}
For $N\ge2$, we could quickly check that $\mathbb{h}[\ket{\psi}^{\text{NOON}}]=0$, for the same reason as above: there is no coherence between adjacent local-number sectors. There is also no known deterministic conversion of a NOON state to single-particle entanglement, nor other useful resource state for quantum telescopy.
\blk
\subsection{Correlated coherent states}
\label{appendixDF}
In contrast to the previous subsection (where a nonlocal state fails to supply a PR resource), here we give an example of a state that does supply a PR resource but is not entangled. This example can be used to expose an FI–QFI gap for quantum telescopy under local operations and classical communications (LOCC).
\par 
\blk
A correlated coherent state—long known as a classical PR resource—is formally defined as:
\begin{align}
\ket{\psi}^{\text{Coh}}_a=\ket{\alpha}_A\ket{\alpha}_B=\sum_{n,m=0}^{\infty}\frac{e^{-|\alpha|^2}\alpha^{n+m}}{\sqrt{n!m!}}\ket{n}_A\ket{m}_B
\end{align}
It is straightforward to see that the QFI ratio $\mathbb{h}[\ket{\psi}^{\text{Coh}}_a$] is nonzero. With $f(n,m)=\frac{e^{-|\alpha|^2}\alpha^{n+m}}{\sqrt{n!m!}}$, let $r=|\alpha|^2$ one has
\begin{align}
 \mathbb{h}[\ket{\psi}^{\text{Coh}}_a]=&\sum_{n,m=1}^{\infty}\frac{2}{(n-1)!(m-1)!(n+m)}e^{-2r}r^{(n+m-1)}\notag \\
 \stackrel{s=n+m-2}{=}& e^{-2r}\sum_{s=0}^{\infty}\sum_{q=0}^s\frac{2r^{s+1}}{q!(s-q)!(s+2)}\notag \\
{=}& e^{-2r}\sum_{s=0}^{\infty}\frac{(2r)^{s+1}}{s!(s+2)}\notag 
 =1-\frac{1-e^{-2r}}{2r}\notag \\
 \stackrel{\text{Large $r$}}{\approx}&1-\frac{1}{2r}=1-\frac{1}{\langle N\rangle}
\end{align}
where in the third equality, we use the identity $\sum_{q=0}^s\frac{s!}{q!(s-q)!}=2^s$, and in the fourth equality we use the identity $\frac{1}{s+2}=\int_0^1t^{s+1}dt$, we have $\sum_{s=0}^{\infty}\frac{(2r)^s}{s!(s+2)}=\int_{0}^1dt 2rt\sum_{s=0}^{\infty}\frac{(2rt)^s}{s!}=\int_{0}^1dt 2rte^{2tr}={e^{2r}}+\frac{1-e^{2r}}{(2r)}$.  

Importantly, this does not contradict the fact that coherent states with homodyne detection achieve FI scaling only as $O(\epsilon^2)$\cite{tsang2011} (i.e., a FI ratio scales $O(\epsilon)$). Instead, it illustrates that LOCC cannot always saturate the QFI in genuinely nonlocal estimation tasks, where nonlocal measurement is required. 
\par 
Thus, to be useful for quantum telescopy, a shared ancilla should be checked for both properties: (i) PR resource and (ii) nonlocal (entangled) resource. Entanglement is crucial to enable effective nonlocal measurements that attain the QFI. Effectively, the entanglement (nonlocality) in the state is essential to implement a nonlocal measurement that attains the QFI. 

Similar to what has already been proved in \cite{tsang2011}, we can prove:
\begin{proposition}
 If $[\rho_a^{AB}]^{T_A}>0$ (positive partial transpose), the FI ratio (instead of QFI ratio) will always go to 0.
\end{proposition}
\begin{proof}
When an ancilla state $\rho_A$ is provided, the most general POVM that one could implement on the source can be written as a convex combination of the ones below. 
\begin{align}
&E(y)=\tr_a[U_{sa}^{A}(y)\otimes U_{sa}^B(y)\mc{E}_{\text{l-ssr}}(\rho^{AB})U_{sa}^{A\dagger}(y)\otimes U_{sa}^{B\dagger}(y)] \notag \\
&\mc{E}_{\text{l-ssr}}(\rho^{AB})=\sum_{n, m=0}^{\infty}(P^A_{n}\otimes P^B_{m})\rho_s^{AB}\otimes\rho_a^{AB} (P^A_{n}\otimes P^B_{m}) \notag
\end{align}
Where $E(y)$ is the measurement operator with outcome $y$, and $U_{sa}^{A/B}(y)$ are local unitary at telescope $A$ and $B$\par 
An immediate consequence of $[\rho_a^{AB}]^{T_A}>0$ is  $E(y)^{T_A}>0$, and following what has been proved in \cite{tsang2011}, if a measurement operator is positive partial transpose, we have the Fisher information ratio. 
\begin{equation}
  \mathbb{f}=O(\epsilon)  \stackrel{\epsilon\rightarrow 0}{\Longrightarrow} 0 
\end{equation}
\end{proof}

Takeaway. A state may be a good PR resource $\mathbb{h} > 0$ and yet fail to deliver comparable FI under local measurements. To achieve the QFI under the local SSR, one typically needs both a PR resource and entanglement (nonlocality) to enable effectively nonlocal measurements. The correlated coherent state exemplifies this gap.

\section{Implementation with linear-optical quantum teleportation}
\label{appendixE}
In this section, we provide further details on the implementation of several entanglement-assisted quantum telescopy protocols. The main ideas for the linear-optical implementations are adapted from earlier work on linear-optical quantum computation \cite{Knill2001, Kok2007}. However, because these specific protocols have not been explicitly laid out for our estimation objective, and because we focus on maximizing Fisher information (FI) for telescopy, we detail the implementation with an emphasis on aspects relevant to FI/QFI in our quantum telescopy task.
\begin{figure}[t]
    \centering
\includegraphics[width=0.5\textwidth]{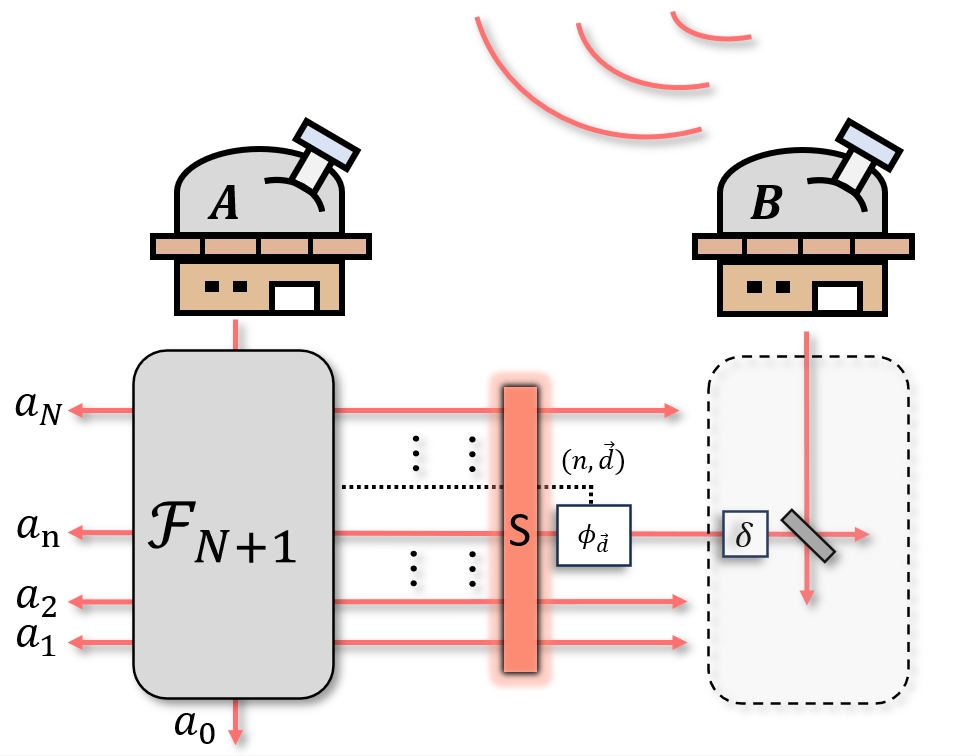}
    \caption{Schematic of near-deterministic teleportation: (1) Quantum Fourier transformation $\mc{F}_{N+1}$ is implemented at telescope A; (2) Measurement outcome $n$: number of photons detected, and $\vec{s}$ arrangement of detection is sent to telescope B for phase correction; (3) A scanning interferometric measurement is performed as in the direct interference scheme. } 
    \label{fig:telescope}
\end{figure}
\par 
Here, we assume the source state is pure for simplicity with $\ket{\psi}_s^{(1)}=\frac{1}{\sqrt{2}}(\ket{0}_A\ket{1}_B+e^{i\theta}\ket{1}_A\ket{0}_B)$, and the case for mixed source state follows straightforwardly by linearity.  Given a source state of the form of Eq.~\ref{eq:LOQC source}, we have the composite system (before applying the superselection rule) as:
\begin{align}
\ket{\psi}^{AB}&=\frac{1}{\sqrt{2}}(\ket{0}_A\ket{1}_B+e^{i\theta}\ket{1}_A\ket{0}_B)\notag \\
&\otimes\sum_{n=0}^Nf_{n}\ket{1}^{\otimes n}_A\ket{0}^{\otimes N-n}_A\ket{0}^{\otimes n}_B\ket{1}^{\otimes N-n}_B   
\end{align}
In the presence of the superselection rule, we have the joint state $\rho^{AB}=\oplus_{n=0}^{N+1}\rho_n^{AB}$, where $\rho_n^{AB}= p_n\op{\psi}{\psi}_n$ with 
\begin{align}
    p_n&=\frac{|f_n|^2}{2}  &n=0, N+1 \\
    p_n&=\frac{{|f_n|^2+|f_{n-1}|^2}}{2} &1\le n\le N
\label{eq:p_n}
\end{align}
Here $p_n$ is the probability weight of the block with $n$ photons at A. The corresponding normalized states are
\begin{widetext}
\begin{align}
  \ket{\psi}_0^{AB}&=\ket{0}^{N+1}_A\ket{1}_B^{N+1}\quad\quad\text{and} \quad\quad  \ket{\psi}_{N+1}^{AB}=\ket{1}^{N+1}_A\ket{0}_B^{N+1} \notag\\
  \ket{\psi}_n^{AB}&=\frac{1}{\sqrt{|f_n|^2+|f_{n-1}|^2}}(f_n\ket{0}_A\ket{1}^{\otimes n}_A\ket{0}^{\otimes N-n}_A\ket{1}_B\ket{0}^{\otimes n}_B\ket{1}^{\otimes N-n}_B+e^{i\theta}f_{n-1}\ket{1}^{\otimes n}_A\ket{0}^{\otimes N-n+1}_A\ket{0}^{\otimes n}_B\ket{1}^{\otimes N-n+1}_B  ) 
\end{align}
\end{widetext}
The first step of the protocol consists of a quantum Fourier transformation on the $N+1$ mode at telescope A with $\mc{F}_{N+1}(a^{\dagger}_p)=\frac{1}{\sqrt{N+1}}\sum_qw^{pq}a^{\dagger}_q=\frac{1}{\sqrt{N+1}}\sum_q\exp[2\pi i\frac{pq}{N+1}]a^{\dagger}_q$ for $p,q \in [0,\cdots, N]$\blk (assuming indistinguishable of the photons). For $n\in [1,\cdots, N-1]$:
\begin{widetext}
\begin{align}
\ket{\psi}_n^{AB}&=\frac{1}{\sqrt{|f_n|^2+|f_{n-1}|^2}}\left(f_n \prod_{i=1}^n a^{\dagger}_ib^{\dagger}_0\prod_{j=n+1}^{N} b^{\dagger}_j+e^{i\theta}f_{n-1}\prod_{i=0}^{n-1} a^{\dagger}_i\prod_{j=n}^{N} b^{\dagger}_j\right)\ket{\text{\textbf{0}}}^{AB} \notag \\
\mc{F}_{N+1}\ket{\psi}_n^{AB}&=\frac{1}{\sqrt{(N+1)^{{n}}(|f_n|^2+|f_{n-1}|^2)}}\left(f_n \prod_{i=1}^n \sum_{q=0}^N \omega^{iq} a^{\dagger}_q b^{\dagger}_0\prod_{j=n+1}^{N} b^{\dagger}_j+e^{i\theta}f_{n-1}\prod_{i=0}^{n-1}\sum_{q=0}^N\omega^{iq} a^{\dagger}_q  \prod_{j=n}^{N} b^{\dagger}_j\right)\ket{\text{\textbf{0}}}^{AB}   
\label{eq: F-output}
\end{align}
\end{widetext}
A certain arrangement $\vec{s}$ is a vector of length $N+1$, representing how many photons are detected at each output port labeled from $[0,\cdots, N]$. As $\vec{s}=[s_0,\cdots, s_N]$ with $\sum_{i=0}^Ns_i=n$. Since both terms in Eq.~\ref{eq: F-output} have contribution to a certain arrangement, there is a phase differences to be fixed, more specifically, the for an arrangement $\vec{s}=[s_0,\cdots, s_N]$ corresponds to a projection to $\ket{\Psi(\vec{s})}^A=\frac{1}{\sqrt{\prod_i s_i!}}  \prod_{i=0}^N (a^{\dagger}_i)^{s_i}\ket{\text{\textbf{0}}}^A$\blk, we have:
\begin{widetext}
\begin{align}
\prescript{A}{}{\bra{\Psi(\vec{s})}}\mc{F}_{N+1}\ket{\psi}_n^{AB}&=\frac{1}{\sqrt{(|f_n|^2+|f_{n-1}|^2)}} \left(f_n c_0 b^{\dagger}_0\prod_{j=k+1}^{N} b^{\dagger}_j+e^{i\theta}f_{n-1}c_1 \prod_{j=k}^{N} b^{\dagger}_j\right)\ket{\text{\textbf{0}}}^{B}  
\end{align}
\end{widetext}
with $ c_0=\frac{1}{\sqrt{(N+1)^{{n}}\prod_n s_n!}}\sum_{\pi\in P^0_{n}}\prod_{i=1}^n\omega^{{d}_i(\vec{s})\pi(i)} $ and $c_1=\frac{1}{\sqrt{(N+1)^{{n}}\prod_n s_n!}}\sum_{\pi\in P^1_{n}}\prod_{i=0}^{n-1}\omega^{{d}_i(\vec{s})\pi(i)}$, 
where $\pi$ is element in permutation group $P^{0}_n$ or $P^{1}_n$, which are permutation acting on the modes $[1,\cdots, n]$, or modes $[0,\cdots, n-1]$ (As the input index set is shifted by $+1$ between the two terms)

Given an arrangement $\vec{s}$, a vector $\vec{d}$ is defined as a length-$n$ vector that specify each particle’s output port\cite{Ticky2010}, i.e., $\vec{s}=[2,1,0]\Rightarrow \vec{d}(\vec{s})=[1,1,2]$, it is not difficult to see that, 
\begin{equation}
    c_0=c_1\omega^{\sum_i d_i(\vec{s})}=c_1\exp[2\pi i\frac{\sum_i{d}_i(\vec{s})}{N+1}]
\end{equation}
Therefore, it is sufficient to introduce a conditional phase gate $\mc{P}_{\phi(\vec{s})}$ with $\phi(\vec{s}) = 2\pi \frac{\sum_i d_i(\vec{s})}{N+1}$ to correct the relative phase. This allows us to obtain phase-corrected states that can be coherently combined, yielding
\begin{align}
&\ket{\psi}_n^B=\sum_{\vec{s}}\mc{P}_{\phi(\vec{s})}\prescript{A}{}{\bra{\Psi(\vec{s})}}\mc{F}_{N+1}\ket{\psi}_n^{AB}\\
&=\frac{1}{\sqrt{|f_n|^2+|f_{n-1}|^2}} \left(f_n b^{\dagger}_0\prod_{j=n+1}^{N} b^{\dagger}_j+e^{i\theta}f_{n-1}\prod_{j=n}^{N} b^{\dagger}_j\right)\ket{\text{\textbf{0}}}^{B} \notag
\end{align}
Where, intuitively, now the bipartite state is teleported to Bob\blk

Then, simply by tracing out all other modes but keeping only the $0$-th and $n$-th modes at Bob's side, with $n$ being the total number of photons detected at telescope A, we have
\begin{equation}
\ket{\psi}_n^B=\frac{1}{\sqrt{|f_n|^2+|f_{n-1}|^2}} \left(f_n b^{\dagger}_0+e^{i\theta}f_{n-1} b^{\dagger}_n\right)\ket{\text{\textbf{0}}}^{B},
\end{equation}
by linearity, assuming a mixed source state $\rho_s^{(1)}$ of the form of Eq.~\ref{eq: rho_1} with $g=|g|e^{i\theta}$, we have:
\begin{equation}
 \rho_n^{B}=\frac{1}{|f_n|^2+|f_{n-1}|^2}\begin{pmatrix}
|f_n|^2 & f_nf^*_{n-1}g \\
 f^*_nf_{n-1}g^* & |f_{n-1}|^2
\end{pmatrix},
\end{equation}
Written in $\{\ket{0}_B\ket{1_n}_B, \ket{1}_B\ket{0_n}_B\}$ basis.

Since the state has now been teleported to Bob, there exists a local measurement whose Fisher information equals the quantum Fisher information (QFI), and the QFI for the state above is proportional to that of Eq. (ref: QFI-epsilon-app). Moreover, the optimal measurement for each parameter is related to its symmetric logarithmic derivative (SLD).

For completeness, we provide the explicit optimal measurements for estimating $|g|$ and $\theta$. Both take the form 
\begin{subequations}
\begin{align}
& \ket{+}_n=(\cos\frac{\alpha_n}{2}b_0^{\dagger}+\sin\frac{\alpha_n}{2}e^{i\delta}b_n^{\dagger})\ket{\text{\textbf{0}}}^{B} \\
& \ket{-}_n=(\sin\frac{\alpha_n}{2}b_0^{\dagger}-\cos\frac{\alpha_n}{2}e^{i\delta}b_n^{\dagger})\ket{\text{\textbf{0}}}^{B}
\end{align}
\end{subequations}
which gives
\begin{widetext}
\begin{align}
P^n_{\pm}&= \epsilon p_n\bra{\pm}\rho^B_n\ket{\pm}_n=\frac{\epsilon p_n}{2}\left[1 \pm \left(\frac{|f_n|^2-|f_{n-1}|^2}{|f_n|^2+|f_{n-1}|^2}\cos\alpha_n+\frac{2|f_n||f_{n-1}||g|}{|f_n|^2+|f_{n-1}|^2}\sin\alpha_n\cos(\theta+\phi-\delta)\right)\right] \notag 
\end{align}
\end{widetext}
where $\epsilon p_n$ is the probability of obtaining the conditional state $\rho^B_n$, with $p_n$ giving in Eq.~\ref{eq:p_n}. And $\phi=\arg(f_nf^{*}_{m-1})$.

\noindent \textit{Optimal measurement for estimating $\theta$}, the optimal measurement for estimating $\theta$:
\begin{align}
  \delta=\theta+\phi+\frac{\pi}{2}  \quad\quad \alpha_n=\frac{\pi}{2}
\end{align}
\textit{Optimal measurement for estimating $|g|$}, the optimal measurement for estimating $\theta$:
\begin{align}
  \delta=\theta+\phi,\quad\tan\alpha_n=\frac{2|f_n||f_{n-1}|}{(|f_n|^2-|f_{n-1}^2|)|g|}
\end{align}
\blk
One can hence evaluate the Fisher information for the corresponding scheme as:
\begin{widetext}
\begin{align}
\mbb{F}_{|g|}&=\sum_{n,\pm}\frac{1}{P^n_{\pm}}\left(\frac{\partial P^n_{\pm} }{\partial |g|}\right)^2\xrightarrow{ \delta=\theta+\phi+\frac{\pi}{2},~~\alpha_n=\frac{\pi}{2},~~\epsilon\ll 1} \sum_{n=1}^N\frac{2|f_n|^2|f_{n-1}|^2}{|f_n|^2+|f_{n-1}|^2}\frac{\epsilon}{1-|g|^2}\\
\mbb{F}_{|\theta|}&=\sum_{n,\pm}\frac{1}{P^n_{\pm}}\left(\frac{\partial P^n_{\pm} }{\partial\theta}\right)^2\xrightarrow{ \delta=\theta+\phi,~~\tan\alpha_n=\frac{2|f_n||f_{n-1}|}{(|f_n|^2-|f_{n-1}^2|)|g|},~~\epsilon\ll 1} \sum_{n=1}^N\frac{2|f_n|^2|f_{n-1}|^2}{|f_n|^2+|f_{n-1}|^2}|g|^2\epsilon\\
\mathbb{f}&=\sum_{n=1}^N\frac{2|f_n|^2|f_{n-1}|^2}{|f_n|^2+|f_{n-1}|^2},
\end{align}
\end{widetext}
 which saturates the Quantum Fisher information we described in Proposition~\ref{prop: assisted-QFI} for any state of the form:
 $$\ket{\psi}_a=\sum_{n=0}^Nf_{n}\ket{1}^{\otimes n}_A\ket{0}^{\otimes N-n}_A\ket{0}^{\otimes n}_B\ket{1}^{\otimes N-n}_B$$
Therefore, the KLM, modified-KLM, and optimal-KLM schemes we introduced are all implementable using local operations and classical communication (LOCC) with linear optics.

\end{document}